\documentclass[12pt]{article}

\pdfoutput=1
\usepackage[margin=1in]{geometry}
\usepackage[T1]{fontenc}
\usepackage{mathpazo}
\usepackage{eucal}
\usepackage{dsfont}
\usepackage{amsmath}
\usepackage{amssymb}
\usepackage{amsthm}
\usepackage{mathtools}
\usepackage{authblk}
\usepackage{bbm}
\usepackage{sepfootnotes}
\usepackage{times}

\usepackage{titlesec}
\usepackage{pifont}
\usepackage{pdfpages}
\usepackage{hyperref}
\usepackage[capitalize,noabbrev]{cleveref}

\usepackage{booktabs}
\usepackage{algorithmic}
\usepackage{algorithm}
\usepackage{optidef}

\theoremstyle{plain}
\newtheorem{thm}{Theorem}[section]
\newtheorem{prop}[thm]{Proposition}
\newtheorem{lmm}[thm]{Lemma}
\newtheorem{cor}[thm]{Corollary}
\newtheorem{obs}[thm]{Observation}
\theoremstyle{definition}
\newtheorem{defi}[thm]{Definition}
\theoremstyle{remark}
\newtheorem{rmk}[thm]{Remark}

\hypersetup{%
  pdfpagemode=UseNone,
  colorlinks=true,
  citecolor=blue,
  linkcolor=blue,
  urlcolor=blue
}

\titleformat{\section}[hang]{\Large\bfseries\filright}{\thesection.}{.5em}{}
\titleformat{\subsection}[hang]{\large\bfseries\filright}{%
  \thesubsection.}{.5em}{}

\newcommand\argmin{\mathop{\mathrm{arg\,min}}\limits} 
\newcommand{\norm}[1]{\ensuremath{\left\lVert #1 \right\rVert}} 
\newcommand{\abs}[1]{\ensuremath{\left\lvert #1 \right\rvert}} 
\newcommand\mc{\mathbb{C}} 
\newcommand\mr{\mathbb{R}} 
 
\newcommand{\eps}{\varepsilon} 

\newcommand\ket[1]{\ensuremath{|#1\rangle}} 
\newcommand\bra[1]{\ensuremath{\langle#1|}} 
\newcommand\bigket[1]{\ensuremath{\bigl|#1\bigr\rangle}}

\newcommand\Ip[2]{\ensuremath{\left\langle#1,#2\right\rangle}} % Inner product

\newcommand\tr{\mathop{\mathrm{tr}}\nolimits}

\newcommand{\set}[1]{\left\{#1\right\}}
\newcommand{\trans}{\intercal}
\newcommand{\mcal}[1]{\mathcal{#1}}
\newcommand{\mbf}[1]{\mathbf{#1}}

\renewcommand{\Pr}[1]{{\textbf{\textup{Pr}} \left[#1\right]}}
\newcommand{\E}{{\mathbb{E}}}

\DeclarePairedDelimiter\rbra{\lparen}{\rparen}
\DeclarePairedDelimiter\sbra{\lbrack}{\rbrack}
\DeclarePairedDelimiter\cbra{\{}{\}}
\DeclarePairedDelimiter\Abs{\lVert}{\rVert}
\newcommand{\poly}{{\operatorname{poly}}}
\newcommand{\polylog}{{\operatorname{polylog}}}
\newcommand{\braket}[2]{\left< #1 \vphantom{#2} \middle| #2 \vphantom{#1} \right>}
\newcommand{\diag}{{\operatorname{diag}}}

\newcommand{\G}{{\mathrm{G}}}

\newcommand{\Oracle}{\ensuremath{\mathcal{O}^\textup{Gibbs}}}
\newcommand{\GibbsDis}[1]{\ensuremath{\frac{\exp{\left(#1\right)}}{\norm{\exp{\left(#1\right)}}_1}}}

\begin{document}

\title{\textbf{\Large Logarithmic-Regret Quantum Learning Algorithms for Zero-Sum Games}}
\author[1,2]{Minbo Gao}
\author[3]{Zhengfeng Ji}
\author[4,5]{Tongyang Li}
\author[6]{Qisheng Wang}
\affil[1]{State Key Laboratory of Computer Science, Institute of Software,\protect\\ Chinese Academy of Sciences, Beijing, China}
\affil[2]{University of Chinese Academy of Sciences, Beijing, China}
\affil[3]{Department of Computer Science and Technology, Tsinghua University, Beijing, China}
\affil[4]{Center on Frontiers of Computing Studies, Peking University, Beijing, China}
\affil[5]{School of Computer Science, Peking University, Beijing, China}
\affil[6]{Graduate School of Mathematics, Nagoya University, Nagoya, Japan
}

\renewcommand\Affilfont{\normalsize\itshape}
\renewcommand\Authfont{\large}
\setlength{\affilsep}{4mm}
\renewcommand\Authand{\rule{6mm}{0mm}}

\date{}

\maketitle

\begin{abstract}
  We propose the first online quantum algorithm for solving zero-sum games with
  $\widetilde O\rbra{1}$ regret under the game setting.\hyperlink{ft1}{\footnotemark[1]}
  Moreover, our quantum algorithm computes an $\eps$-approximate Nash
  equilibrium of an $m \times n$ matrix zero-sum game in quantum time
  $\widetilde O\rbra{\sqrt{m+n}/\eps^{2.5}}$.
  Our algorithm uses standard quantum inputs and generates classical outputs
  with succinct descriptions, facilitating end-to-end applications.
  Technically, our online quantum algorithm ``quantizes'' classical algorithms
  based on the \emph{optimistic} multiplicative weight update method.
  At the heart of our algorithm is a fast quantum multi-sampling procedure for
  the Gibbs sampling problem, which may be of independent interest.
\end{abstract}

\hypertarget{ft1}{\footnotetext[1]{\label{ft1}Throughout this
  paper, $\widetilde{O}\rbra{\cdot}$ suppresses polylogarithmic factors such
  as $\log\rbra{n}$ and $\log\rbra{1/\eps}$, and $O^*\rbra{\cdot}$ hides
  quasi-polylogarithmic factors such as $n^{o\rbra{1}}$ and
  $\rbra{1/\eps}^{o\rbra{1}}$.}}

  \section{Introduction}

  Nash equilibrium is one of the most important concepts in game theory.
  It characterizes and predicts rational agents' behaviors in non-cooperative
  games, finding a vast host of applications ranging from analyzing
  wars~\cite{S1980} and designing auctions~\cite{M1981}, to optimizing
  networks~\cite{RT2002}.
  
  It was shown in~\cite{DGP2009,CD2006} that finding a Nash equilibrium is
  $\mathsf{PPAD}$-hard for general games.
  Nevertheless, computing the Nash equilibrium for specific types of games, such as
  zero-sum games, is particularly interesting.
  A zero-sum game requires that the utility of one player is the opposite of the
  other's, a condition that often appears in, for example, chess games.
  Von Neumann's minimax theorem~\cite{vN28} promises that every finite two-player
  zero-sum game has optimal mixed strategies.
  
  \paragraph{Zero-Sum Games.} For a two-player zero-sum game represented by an
  $m \times n$ matrix $\mbf{A}$, the Nash equilibrium is the solution pair
  $\rbra{x, y}$ to the following min-max problem:
  \begin{equation*}
      \min_{x\in \Delta_m} \max_{y \in \Delta_n} x^{\trans} \mbf{A} y,
  \end{equation*}
  where $x \in \Delta_m$ and $y \in \Delta_n$ are $m$- and $n$-dimensional
  probability distributions, respectively.
  Usually, we are satisfied with an approximate Nash equilibrium rather than an
  exact one.
  An $\eps$-approximate Nash equilibrium is a solution pair $\rbra{x, y}$ such
  that:
  \begin{equation*}
    \max_{y'\in \Delta_n} x^\intercal \mathbf{A} y' -
    \min_{x'\in \Delta_m} x'^\intercal\mathbf{A} y \le \eps.
  \end{equation*}
  
  \paragraph{Online Learning.} Since the matrix $\mbf{A}$ of the zero-sum game usually has a large dimension in practice, it is common that we trade accuracy for space and time efficiency.
  Thus, online learning becomes increasingly
  significant in these scenarios.
  Online learning studies the situation when data is only available in sequential
  order and aims at making good decisions in this setup.
  In evaluating online learning algorithms, regret is an important criterion
  that measures how good an algorithm is compared with the optimal static loss
  (see more details in \cref{sec:online}).
  
  The idea of the online learning algorithms for zero-sum games stems from
  repeated play in game theory, e.g., fictitious play~\cite{CNL2006}.
  Specifically, we simulate the actions of two players for multiple rounds. 
  In each round, players make decisions using a no-regret learning algorithm,
  considering the opponent's previous actions.
  For example, a famous algorithm of this type was proposed in~\cite{GK1995}
  inspired by the exponential Hedge algorithm.
  The algorithm has regret $\widetilde O \rbra{\sqrt{T}}$ and $T$ rounds,
  establishing the convergence rate of $\widetilde O \rbra{1/\sqrt{T}}$.
  
  It takes about two decades before the next improvement in~\cite{DDK2015} to
  happen, where the authors proposed a strongly-uncoupled algorithm, achieving
  $\widetilde O\rbra{1}$ total regret if both players use the algorithm.
  They used the technique of minimizing non-smooth functions using smoothed approximations proposed in~\cite{Yu05a}, and this technique was later developed in~\cite{Yu05b, A04} for broader classes of problems.
  Later, it was found in~\cite{SALS2015} that the optimistic multiplicative
  weight algorithm also leads to $\widetilde O\rbra{1}$ total regret with regret
  bounded by variation in utilities; this algorithm was recently extended to
  correlated equilibria in multi-player general-sum games in~\cite{ADFFGS2022}.
  It was proved in~\cite{HAM2021} that optimistic mirror descent with
  a time-varying learning rate can also achieve $\widetilde O \rbra{1}$ total regret
  for multi-players.
  Our quantum algorithm follows the overall idea of the optimistic multiplicative
  weight update and the regret bounded by variation methods~\cite{SALS2015}.
  
  \paragraph{Quantum Computing and Learning.}
  Quantum computing has been rapidly advancing in recent years.
  Specifically, many machine learning problems are known to have significant
  quantum speedups, e.g., support vector machines~\cite{RML2014}, principal
  component analysis~\cite{LMR2014}, classification~\cite{KWS2016,LCW2019}, etc.
  The combination of quantum computing and online learning has
  recently become a popular topic.
  For instance, online learning tools have been applied to solving semidefinite
  programs (SDPs) with quantum speedup in the problem dimension and the number of
  constraints~\cite{BS2017,vGGd2017,BKLLSW2019,vG2019}.
  In addition, as an important quantum information task, the online version of
  quantum state learning has been systematically developed with good theoretical
  and empirical guarantees~\cite{ACHKN2018,YJZS2020,CW2020,CHLLWY2022}.
  
  For finding the Nash equilibrium of zero-sum games, a quantum algorithm was
  proposed in~\cite{vAG2019} by ``quantizing'' the classical algorithm
  in~\cite{GK1995}, achieving a quadratic speedup in the dimension parameters $m$
  and $n$.
  At the same time, quantum algorithms for training linear and kernel-based
  classifiers were proposed in~\cite{LCW2019}, which have similar problem
  formulations to zero-sum games.
  Recently, an improved quantum algorithm for zero-sum games was proposed
  in~\cite{BGJST2022} using dynamic Gibbs sampling.
  All of the above quantum algorithms are based on the multiplicative weight update method, 
  and as a consequence, they all share the $O\rbra{\sqrt{T}}$ regret bound.

  \begin{table*}[!htp]
    \begin{center}\footnotesize
    \caption{Online Algorithms for  $\eps$-Approximate Nash Equilibria of Zero-Sum Games.}\label{tab:main}
    \begin{tabular}{lllll}
      \toprule
      Approach & Type & Regret & Update Cost Per Round
      & Classical/Quantum Time Complexity \\
      \midrule
      \cite{GK1995} & Classical & $\widetilde O\rbra{\sqrt{T}}$
      & $\widetilde O\rbra{m+n}$ & $\widetilde O\rbra{\rbra{m+n}/\eps^2 }$ \\
      \cite{vAG2019} & Quantum & $\widetilde O\rbra{\sqrt{T}}$
      & $\widetilde O\rbra{\sqrt{m+n}/\eps}$
      & $\widetilde O\rbra{\sqrt{m+n}/\eps^3}$ \\
      \cite{BGJST2022} & Quantum & $\widetilde O\rbra{\sqrt{T}}$
      & $\widetilde O\rbra{\sqrt{m+n}/\eps^{0.5}}$
      & $\widetilde O\rbra{\sqrt{m+n}/\eps^{2.5}}~\hyperlink{ft2}{\footnotemark[2]}$\\
      \midrule
      \cite{SALS2015} & Classical & $\widetilde O \rbra{1}$
      & $\widetilde O\rbra{mn}$ & $\widetilde O\rbra{mn/\eps}$ \\
      \cite{CJST2019} & Classical & $\widetilde O \rbra{1}$
      & $\widetilde O\rbra{\sqrt{mn\rbra{m+n}}}$
      & $\widetilde O\rbra{mn+\sqrt{mn\rbra{m+n}}/ \eps}$ \\
      Our result~\hyperlink{ft3}{\footnotemark[3]}  & Quantum & $\widetilde O \rbra*{1}$
      & $\widetilde O \rbra{\sqrt{m+n}/\eps^{1.5}}$
      & $\widetilde O\rbra{\sqrt{m+n}/\eps^{2.5}}$ \\
      \bottomrule
    \end{tabular}

    \end{center}
  \end{table*}

\hypertarget{ft2}{\footnotetext[2]{The complexity given in \cite{BGJST2022} is $\widetilde O (\sqrt{m+n}/\eps^{2.5}+1/\eps^3)$, wherein the former term $\sqrt{m+n}/\eps^{2.5}$ dominates the complexity. See Footnote \ref{ft4} and \cref{rem:term} for discussions.}}
\hypertarget{ft3}{\footnotetext[3]{Here, we require that the input matrix $\mbf{A}$ satisfies $\Abs{\mbf{A}}\le 1$, while other works require $|A_{i,j}|\le 1$.}} 
  \subsection{Main Result}
  Our result in this paper establishes a positive answer to the following open question:
  \emph{Does there exist a learning algorithm with $\widetilde O\rbra{1}$ regret
    allowing quantum speedups?}
  
  Inspired by the optimistic follow-the-regularized-leader algorithm proposed
  in~\cite{SALS2015}, we propose a sample-based quantum online learning algorithm
  for zero-sum games with $O\rbra{\log\rbra{mn}}$ total regret, which is
  near-optimal.
  If we run this algorithm for $T$ rounds, it will compute an
  $\widetilde O\rbra{1/T}$-approximate Nash equilibrium with high probability,
  achieving a quadratic speedup in dimension parameters $m$ and $n$.
  Formally, we have the following quantum online learning algorithm:
  
  \begin{thm}[Online learning for zero-sum games]
    Suppose $T\leq\widetilde O\rbra{m+n}$.
    There is a quantum online algorithm for zero-sum game
    $\mbf{A} \in \mr^{m \times n}$ with $\Abs{\mbf{A}} \leq 1$ such that it
    achieves a total regret of $O\rbra{\log\rbra{mn}}$ with high probability after
    $T$ rounds, while each round takes quantum time
    $\widetilde O\rbra{T^{1.5}\sqrt{m+n}}$.
  \end{thm}
  
  Our algorithm does not need to read all the entries of the input matrix
  $\mbf{A}$ at once.
  Instead, we assume that our algorithm can query its entries when necessary.
  The input model is described as follows:
  \begin{itemize}
    \item Classically, given any $i\in [m], j \in [n]$, the entry $A_{i,j}$ can be
          accessed in $\widetilde O(1)$ time.
    \item Quantumly, we assume that the entry $A_{i,j}$ can be accessed in
          $\widetilde O(1)$ time \emph{coherently}.
  \end{itemize}
  This is the \emph{standard quantum input model} for zero-sum games adopted in
  previous literature~\cite{LCW2019,vAG2019,BGJST2022}.
  See more details in \cref{sec:intro-q}.
  
In addition, same as prior works~\cite{vAG2019,BGJST2022},  our algorithm outputs \emph{purely classical vectors} with succinct
  descriptions because they are sparse (with at most $T^2$ nonzero entries).
  Overall, using standard quantum inputs and generating classical outputs
  significantly facilitate end-to-end applications of our algorithm in the near
  term.
  
  As a direct corollary, we can find an \emph{$\eps$-approximate Nash equilibrium}
  by taking $T = \widetilde O\rbra{1/\eps}$, resulting in a quantum speedup stated
  as follows.
  A detailed comparison to previous literature is presented in \cref{tab:main}.

  \begin{cor} [Computing Nash equilibrium]\label{cor:main-nash} There is a
    quantum online algorithm for zero-sum game $\mbf{A} \in \mr^{m \times n}$ with
    $\Abs{A} \leq 1$ that, with high probability, computes an
    $\eps$-approximate Nash equilibrium in quantum time
  \footnotetext[4]{\label{ft4}In fact, a condition of
  $\eps = \Omega\rbra{\rbra{m+n}^{-1}}$ is required in our quantum algorithm.
  Nevertheless, our claim still holds because when
  $\eps = O\rbra{\rbra{m+n}^{-1}}$, we can directly apply the classical
  algorithm in~\cite{GK1995} with time complexity
  $\widetilde O \rbra{\rbra*{m+n}/\eps^2}
  \leq \widetilde O \rbra{\sqrt{m+n}/\eps^{2.5}}$.}
    $\widetilde O \rbra{\sqrt{m+n}/\eps^{2.5}}$.\hyperref[ft4]{\footnotemark[4]}
  \end{cor}

  \paragraph{Quantum Lower Bounds.}
  In the full version of~\cite{LCW2019}, they showed a lower bound
  $\Omega\rbra{\sqrt{m+n}}$ for the quantum query complexity
  of computing an $\eps$-approximate
  Nash equilibrium of zero-sum games for constant $\eps$ .
  Therefore, our algorithm is tight in terms of $m$ and $n$.
  \subsection{Our Techniques}
  
  Our quantum online algorithm is a stochastic modification of the optimistic
  multiplicative weight update proposed by~\cite{SALS2015}.
  We choose to ``quantize'' the optimistic online algorithm because it has a
  better convergence rate for zero-sum games than general multiplicative weight
  update algorithms.
  During the update of classical online algorithms, the update term (gradient
  vector) is computed in linear time by arithmetic operations.
  However, we observe that it is not necessary to know the exact gradient vector.
  This motivates us to apply stochastic gradient methods for updates so that certain
  fast quantum samplers can be utilized here and bring quantum speedups.
  
  Specifically, in our quantum online algorithm, we need to establish an upper
  bound on the expectation of the total regret of our stochastic update rule, and
  also deal with errors that may appear from noisy samplers.
  To reduce the variance of the gradient estimator, we need multiple samples (from
  Gibbs distributions) at a time.
  To this end, we develop a fast quantum multi-Gibbs sampler that produces
  multiple samples by preparing and measuring quantum states.
  
  \paragraph{Sample-Based Optimistic Multiplicative Weight Update.}
  Optimistic online learning adds a ``prediction loss'' term to the cumulative
  loss for regularized minimization, giving a faster convergence rate than the
  non-optimistic versions for zero-sum games.
  \cite{AHK2012} surveyed the use of the multiplicative weight update method in
  various domains, but little was known for the optimistic learning
  method at that time.
  \cite{DDK2015} proposed an extragradient method that largely resembles the
  optimistic multiplicative weight.
  \cite{SALS2015} gave a characterization of this update rule---RVU (Regret
  bounded by Variation in Utilities) property, which is very useful in proving
  regret bounds.
  Subsequently, the optimistic learning method is applied to other areas, including
  training GANs~\cite{DISZ2018} and multi-agent learning~\cite{PZO2022}.
  
  However, when implementing the optimistic learning methods, we face a
  fundamental difficulty---we cannot directly access data from quantum states
  without measurement.
  To resolve this issue, we get samples from the desired distribution and use them
  to estimate the actual gradient.
  This idea is commonly seen in previous literature on quantum SDP
  solvers~\cite{BS2017,vGGd2017,BKLLSW2019,vG2019}.
  Then we prove the regret bound (see \cref{thm:error-Gibbs}) of our algorithm by
  showing that it has a property similar to the RVU property~\cite{SALS2015}.
  Moreover, we need multiple samples to obtain a small ``variance'' of the
  stochastic gradient (by taking the average of the samples), to ensure that the
  expected regret is bounded.
  Our fast quantum multi-Gibbs sampler produces the required samples and ensures
  further quantum speedups.
  In a nutshell, we give an algorithm (\cref{algo:POMWU}) which modifies the
  optimistic multiplicative weight algorithm in~\cite{SALS2015} to fit the
  quantum implementation.
  This is \emph{the first quantum algorithm that implements optimistic online
    learning} to the best of our knowledge.
  
  \paragraph{Fast Quantum Multi-Gibbs Sampling.}
  The key to our sample-based approach is to obtain multiple samples from the
  Gibbs distribution after a common preprocessing step.
  For a vector $p \in \mr^{n}$ with $\max_{i\in[n]}\abs{p_i}\le \beta$, a sample
  from the Gibbs distribution with respect to $p$ is a random variable $j\in [n]$
  such that
  $
    \Pr{j = l} = \frac{\exp \rbra*{p_l}}{\sum_{i = 1}^n \exp \rbra*{p_i}}.
  $
   Gibbs sampling on a quantum computer was first studied in~\cite{PW2009}, and
   was later used as a subroutine in quantum SDP
   solvers~\cite{BS2017,vGGd2017,BKLLSW2019,vG2019}.
   However, the aforementioned quantum Gibbs samplers produce one sample from an
   $n$-dimensional Gibbs distribution in quantum time
   $\widetilde O\rbra{\beta\sqrt{n}}$; thus, we can produce $k$ samples in quantum
   time $\widetilde O\rbra{\beta k \sqrt{n}}$.
   Inspired by the recent work~\cite{H2022} about preparing multiple samples of
   quantum states, we develop a fast quantum Gibbs sampler (\cref{thm:gibbs})
   which \emph{produces $k$ samples from a Gibbs distribution in quantum time
     $\widetilde O\rbra{\beta \sqrt{nk}}$}.
   Our quantum multi-Gibbs sampling may have potential applications in
   sample-based approaches for optimization tasks that require multiple samples.
  
   Technically, the main idea is based on quantum rejection
   sampling~\cite{G2000,ORR2013}, where the target quantum state $\ket{u}$ is
   obtained by post-selection from a quantum state $\ket{u_{\text{guess}}}$ that
   is easy to prepare (see \cref{sec:multi-Gibbs}).
   The algorithm has the following steps (Here we assume $\beta = O\rbra{1}$ for
   simplicity):
   \begin{itemize}
    \item To bring $\ket{u_{\text{guess}}}$ closer to the target $\ket{u}$, we
          find the $k$ (approximately) most dominant amplitudes of $\ket{u}$ by
          quantum $k$-maximum finding~\cite{DHHM2006} in quantum time
          $\widetilde O\rbra{\sqrt{nk}}$.
    \item In quantum $k$-maximum finding, we need to estimate the amplitudes of
          $\ket{u}$ and compare them coherently, which requires the estimation
          should be consistent.
          To address this issue, we develop consistent quantum amplitude
          estimation (see \cref{append:consistent-amplitude-estimation}) for our
          purpose based on consistent phase estimation~\cite{A2012,T2013}, which
          is of independent interest.
    item Then, we correct the tail amplitudes of $\ket{u_{\text{guess}}}$ by
          quantum singular value transformation~\cite{GSLW2019}, resulting in an
          (approximate) target state with amplitude $\Omega\rbra{\sqrt{k/n}}$ (see
          \cref{append:multi-Gibbs-proof} for details).
    \item Finally, we post-select the target state by quantum amplitude
          amplification~\cite{BHMT02} in quantum time
          $\widetilde O\rbra{\sqrt{n/k}}$.
          This follows that $k$ samples of the target quantum state can be
          obtained in
          $k \cdot \widetilde O\rbra{\sqrt{n/k}} = \widetilde O\rbra{\sqrt{nk}}$
          time.
  \end{itemize} 
  
  We believe that our technique can be extended to a wide range of distributions
  whose mass function is monotonic and satisfies certain Lipschitz conditions.
  
  \section{Preliminaries}

This section introduces some basic notions of game theory, online learning, and
quantum computing.
  \subsection{General Mathematical Notations}
  
  For convenience, we use $[n]$ to denote the set $\set{ 1,2,\dots,n}$.
  \paragraph{Vectors.}
  We use $e_i$ to denote a vector whose $i$-th coordinate is $1$ and other
  coordinates are $0$.
  For a vector $v\in \mr^n$, $v_i$ is the $i$-th coordinate of $v$.
  For a function $f\colon \mr \to \mr$, we write $f\rbra*{v}$ to denote the result
  of $f$ applied to its coordinates, i.e.,
  $f\rbra*{v} = \rbra*{ f\rbra*{v_1}, f\rbra*{v_2},\dots, f\rbra*{v_n}}$.
  \paragraph{Vector Spaces.}
  We use $\Delta_{n}$ to represent the set of $n$-dimensional probability
  distributions, i.e.,
  \[
    \Delta_n := \{v\in \mr^n : \sum_{i=1}^n v_i = 1, \forall i\in [n], v_i\ge 0\}.
  \]
  Here the $i$-th coordinate represents the probability of event $i$ takes place. 
  We use $\norm{\cdot}$ for vector norms.
  The $l_1$ norm $\norm{\cdot}_1$ for a vector $v\in \mr^n$ is defined as
  $\norm{v}_1:=\sum_{i=1}^n \abs{v_i}$.
  For two $n$-dimensional probability distributions $p,q\in \Delta_n$, their total
  variance distance is defined as:
  \[
    d_{\text{TV}}\rbra{p,q} = \frac{1}{2}\norm{p-q}_1 = \frac{1}{2}\sum_{i=1}^n \abs{p_i-q_i}.
  \]
  \paragraph{Matrices.}
  We will use $\mbf{A}\in \mr^{m\times n}$ to denote a matrix
  with $m$ rows and $n$ columns.
  $A_{i,j}$ is the entry of $\mbf{A}$ in the $i$-th row and $j$-th column.
  $\norm{\mbf{A}}$ denotes the operator norm of matrix
  $\mbf{A}$.
  For a vector $v\in \mr^n$, we use $\diag \rbra*{v}$ to denote the diagonal
  matrix in $\mr^{n\times n}$ whose diagonal entries are coordinates of $v$ with
  the same order.
  
  \subsection{Game Theory}
  
  Let us consider two players, Alice and Bob, playing a zero-sum game represented
  by a matrix $\mathbf{A}\in \mr^{m\times n}$ with entries $A_{i,j}$, where $[m]$
  is the labeled action set of Alice and $[n]$ of Bob.
  Usually, Alice is called the row player and Bob is called the column player.
  $A_{i,j}$ is the payoff to Bob and $-A_{i,j}$ is the payoff to Alice when Alice
  plays action $i$ and Bob plays action $j$.
  Both players want to maximize their payoff when they consider their opponent's
  strategy.
  
  Consider the situation that both players' actions are the best responses to each
  other.
  In this case, we call the actions form a \emph{Nash equilibrium}.
  A famous minimax theorem by~\cite{vN28} states that we can exchange the order
  of the min and max operation and thus the value of the Nash equilibrium of the
  game can be properly defined.
  To be more exact, we have:
  $
    \min_{x\in \Delta_m}\max_{y\in \Delta_n} x^{\trans} \mbf{A} y =
    \max_{y\in \Delta_n}\min_{x\in \Delta_m} x^{\trans} \mbf{A} y
  $.
  Here the value of the minimax optimization problem is called the \emph{value of
    the game}. Our goal is to find an approximate Nash equilibrium for this problem.
  More formally, we need two probabilistic vectors $x\in \Delta_m, y\in \Delta_n$
  such that the following holds:
  \[
    \max_{y'\in \Delta_n} x^\intercal\mathbf{A} y' -
    \min_{x'\in \Delta_m} x'^\intercal\mathbf{A} y \le \eps.
  \]
  We will call such a pair $(x,y)$ an $\eps$-\textit{approximate Nash Equilibrium}
  for the two-player zero-sum game $\mathbf{A}$.
  Computing an $\eps$-approximate Nash for a zero-sum game is not as hard as for a
  general game.
  For a general game, approximately computing its Nash equilibrium is
  $\mathsf{PPAD}$-hard~\cite{DGP2009,CD2006}.
  
  Here, we emphasize an important observation that will be used throughout our
  paper: we can add a constant or multiply by a positive constant on all
  $\mbf{A}$'s entries without changing the solution pair of the Nash equilibrium.
  This is because, in the definition of the Nash equilibrium, the best response
  is always in comparison with other possible actions, so only the order of the
  utility matters.
  Thus without loss of generality, we can always rescale $\mbf{A}$ to $\rbra*{\mbf{A}+c\mathbf{1}}/2$ where $\mathbf{1}$ is an $m\times n$
  matrix with all entries being $1$ with $c$ being the largest absolute value of $\mbf{A}$'s entries,
  and let $\mbf{A'} = \mbf{A}/\norm{\mbf{A}}$ to guarantee that $\mbf{A}$ has non-negative entries and has operator norm no more than $1$.
  
  \subsection{Online Learning}\label{sec:online}
  
  \subsubsection{Notions of online learning}
  In general, online learning focuses on making decisions in an uncertain
  situation, where a decider is not aware of the current loss and is required to
  make decisions based on observations of previous losses.
  To be more exact, we fix the number of rounds $T$ and judge the performance of
  the algorithm (in the following of this subsection we use ``decider'' with the
  same meaning as the ``algorithm'') in these $T$ rounds.
  Assume that the decider is allowed to choose actions in the domain $\mcal{X}$,
  usually a convex subset of a finite-dimensional Euclidean space.
  Let $t\in [T]$ denote the current number of rounds.
  At round $t$, the decider chooses an action $x_t\in \mcal{X}$.
  (The action may depend on the decider's previous observation of the loss
  functions $l_i$ for all $i\in [t-1]$.)
  Then the decider will get the loss $l_t(x_t)$ for the current round.
  We assume that the decider can also observe the full information of $l_t$, i.e.,
  the formula of the function.
  We judge the performance of the algorithm by its \emph{regret}:
  \[
     \mcal{R}\rbra*{T} = \sum_{i=1}^T l_t(x_t) - \min_{x\in \mcal{X}} \sum_{i=1}^T l_t(x).
  \]
  Intuitively, the regret shows how far the total loss in $T$ rounds caused by the
  algorithm is from the optimal static loss.
  
  \subsubsection{Online learning in zero-sum games}
  
  To demonstrate how to compute the approximate Nash equilibrium using online
  learning algorithms, we present a useful proposition here. It states that
  from any sublinear regret learning algorithm $\mcal{A}$ with regret
  $\mcal{R}\rbra*{T}$, we can find an $O\rbra*{\mcal{R}\rbra*{T}/T}$-approximate
  Nash equilibrium of the zero-sum game in $T$ rounds.

  To be more precise, let us consider the following procedure.
  Let $\mbf{A}$ be the matrix for the two-player zero-sum game.
  The algorithm starts with some initial strategies
  $u_0\in \Delta_m, v_0\in \Delta_n$ for the two players.
  Then at each round $t$, for each player, it makes decisions with previous
  observations of the opponent's strategy.
  In particular, the row player is required to choose his/her action
  $u_t\in \Delta_m$ after considering the previous loss functions
  $g_i(x) = v_i^{\trans} \mbf{A}^{\trans}x $ for $i\in [t-1]$.
  Similarly, the column player chooses his/her action $v_t\in \Delta_n$ with
  respect to the previous loss functions $h_i(y) = -u_i^{\trans}\mbf{A}y$ for
  $i\in [t-1]$.
  After both players choose their actions at this round $t$, they will receive
  their loss functions $g_t(x) := v_t^{\trans} \mbf{A}^{\trans}x$ and
  $h_t(y) = -u_t^{\trans}\mbf{A}y$, respectively.
  
  Suppose after $T$ rounds, the regret of the row player with respect to the loss
  functions $g_t\rbra*{x}$ is $\mcal{R}\rbra*{T}$, and the regret of the column
  player with loss functions $h_t\rbra*{y}$ is $\mcal{R}'\rbra*{T}$.
  We can write the total regret $\mcal{R}\rbra*{T}+\mcal{R}'\rbra*{T}$ explicitly:
  \[
     \mcal{R}\rbra*{T} +\mcal{R}'\rbra*{T}
    = T \rbra{\max_{v\in\Delta_n} \Ip{v}{\mbf{A}^{\trans}\hat{u}} -
      \min_{u\in \Delta_m}\Ip{u}{\mbf{A}\hat{v}}},
  \]
  where the average strategy is defined as $\hat{u} = \sum_{i=1}^T u_i /T$,
  $\hat{v} = \sum_{i=1}^T v_i/T$.
  This pair is a good approximation of the Nash equilibrium for
  the game $\mbf{A}$ if the regret is $o\rbra{T}$.
  
  \subsection{Quantum Computing}\label{sec:intro-q}
  
  In quantum mechanics, a $d$-dimensional quantum state is described by a unit
  vector
  \begin{equation*}
    v = \rbra{v_0, v_1, \ldots, v_{d-1}}^{\trans},
  \end{equation*}
  usually denoted as $\ket{v}$ with the Dirac symbol $\ket{\cdot}$, in a complex
  Hilbert space $\mc^{d}$.
  The computational basis of $\mc^{d}$ is defined as $\cbra{\ket{i}}_{i=0}^{d-1}$,
  where $\ket{i} = \rbra{0, \dots, 0, 1, 0, \dots, 0}^{\trans}$ with the $i$-th
  (0-indexed) entry being $1$ and other entries being $0$.
  The inner product of quantum states $\ket{v}$ and $\ket{w}$ is defined by
  $\braket{v}{w} = \sum_{i=0}^{d-1} v_i^* w_i$, where $z^*$ denotes the conjugate
  of complex number $z$.
  The norm of $\ket{v}$ is defined by $\Abs{\ket{v}} = \sqrt{\braket{v}{v}}$.
  The tensor product of quantum states $\ket{v} \in \mc^{d_1}$ and
  $\ket{w} \in \mc^{d_2}$ is defined by
  $\ket{v} \otimes \ket{w} = \rbra{v_0w_0, v_0w_1, \dots, v_{d_1-1}w_{d_2-1}}^{\trans}
  \in \mc^{d_1d_2}$,
  denoted as $\ket{v}\ket{w}$ for short.

  A quantum bit (qubit for short) is a quantum state $\ket{\psi}$ in $\mc^{2}$,
  which can be written as $\ket{\psi} = \alpha \ket{0} + \beta \ket{1}$ with
  $\abs{\alpha}^2 + \abs{\beta}^2 = 1$.
  An $n$-qubit quantum state is in the tensor product space of $n$ Hilbert spaces
  $\mc^{2}$, i.e., $\rbra{\mc^{2}}^{\otimes n} = \mc^{2^n}$ with the computational
  basis $\cbra{\ket{0}, \ket{1}, \dots, \ket{2^n-1}}$.
  To obtain classical information from an $n$-qubit quantum state $\ket{v}$, we
  measure $\ket{v}$ on the computational basis and obtain outcome $i$ with
  probability $p\rbra{i} = \abs{\braket{i}{v}}^2$ for every $0 \leq i < 2^n$.
  The evolution of a quantum state $\ket{v}$ is described by a unitary
  transformation $U \colon \ket{v} \mapsto U\ket{v}$ such that
  $UU^\dag = U^\dag U = I$, where $U^\dag$ is the Hermitian conjugate of $U$, and
  $I$ is the identity operator.
  A quantum gate is a unitary transformation that acts on $1$ or $2$ qubits, and a
  quantum circuit is a sequence of quantum gates.
  
  Throughout this paper, we assume the quantum oracle $\mathcal{O}_{\mbf{A}}$ for
  a matrix $\mbf{A} \in \mr^{m \times n}$, which is a unitary operator such that
  for every row index $i \in \sbra*{m}$ and column index $j \in \sbra*{n}$,
  \[
    \mathcal{O}_{\mbf{A}} \ket{i}\ket{j}\ket{0} = \ket{i}\ket{j}\ket{A_{i,j}}.
  \]
  Intuitively, the oracle $\mathcal{O}_{\mbf{A}}$ reads the entry $A_{i,j}$ and
  stores it in the third register; this is potentially stronger than the classical
  counterpart when the query is a linear combination of basis vectors, e.g.,
  $\sum_{k}\alpha_{k}\ket{i_{k}}\ket{j_{k}}$ with $\sum_{k}|\alpha_{k}|^{2}=1$.
  This is known as the \emph{superposition} principle in quantum computing.
  Note that this input model for matrices is commonly used in quantum algorithms,
  e.g., linear system solvers~\cite{HHL2009} and semidefinite programming
  solvers~\cite{BS2017,vGGd2017,BKLLSW2019,vG2019}.
  
  A quantum (query) algorithm $\mathcal{A}$ is a quantum circuit that consists a
  sequence of unitary operators $G_1, G_2, \dots, G_T$, each of which is either a
  quantum gate or a quantum oracle.
  The quantum time complexity of $\mathcal{A}$ is measured by the number $T$ of
  quantum gates and quantum oracles in $\mathcal{A}$.
  The execution of $\mathcal{A}$ on $n$ qubits starts with quantum state
  $\ket{0}^{\otimes n}$, then it performs unitary operators $G_1, G_2, \dots, G_T$
  on the quantum state in this order, resulting in the quantum state
  $\ket{\phi} = G_T \dots G_2 G_1 \ket{0}^{\otimes n}$.
  Finally, we measure $\ket{\phi}$ on the computational basis $\ket{i}$ for
  $0 \leq i < 2^n$, giving a classical output $i$ with probability
  $|\langle i|\phi\rangle|^{2}$.
  
  \section{Quantum Algorithm for Online Zero-Sum Games by Sample-Based Optimistic
    Multiplicative Weight Update}
  
  Now, we present our quantum algorithm for finding an approximate Nash
  equilibrium for zero-sum games (\cref{algo:POMWU}).
  This algorithm is a modification of the optimistic multiplicative weight
  algorithm~\cite{SALS2015}, in which we use stochastic gradients to estimate true
  gradients.
  This modification utilizes the quantum advantage of Gibbs sampling (as will be
  shown in \cref{sec:multi-Gibbs}).
  It is the source of quantum speedups in the algorithm and also the reason that
  we call the algorithm \textit{sample-based}.
  To this end, we first give the definition of Gibbs sampling oracles.
  
  \begin{defi} [Approximate Gibbs sampling oracle]\label{def:Gibbs-oracle}
    Let $p \in \mr^n$ be an $n$-dimensional vector, $\epsilon >0 $ be the approximate
    error.
    We let $\Oracle_{p}(k,\epsilon)$ denote the oracle which produces $k$
    independent samples from a distribution that is $\epsilon$-close to the Gibbs
    distribution with parameter $p$ in total variation distance.
    Here, for a random variable $j$ taking value in $[n]$ following the Gibbs
    distribution with parameter $p$, we have
    $ \Pr{j = l} = \exp \rbra*{p_l}/{\sum_{i = 1}^n \exp \rbra*{p_i}}.$
  \end{defi}
  
  \renewcommand{\algorithmicrequire}{\textbf{Input:}}
  \renewcommand{\algorithmicensure}{\textbf{Output:}}
  
  \begin{algorithm}[!htp]
    \caption{Sample-Based Optimistic Multiplicative Weight Update for Matrix Games}\label{algo:POMWU}
    \begin{algorithmic}[1]
      \REQUIRE $\mathbf{A}\in \mr^{m\times n}$, additive approximation $\eps$,
      approximate Gibbs sampling oracle $\Oracle$ with error
      $\epsilon_{\G}$, total episode $T$, learning rate
      $\lambda \in \rbra{0, \sqrt{3}/6}$.
      \ENSURE $\left(\hat{u},\hat{v}\right)$ as the approximate Nash equilibrium
      of the matrix game $\mbf{A}$.
      \STATE Initialize $\hat{u} \gets \mathbf{0}_m$,
      $\hat{v} \gets \mathbf{0}_n$, $x_1\gets \mathbf{0}_m$,
      $y_1\gets \mathbf{0}_n$.
      \STATE Set $g_{1} \gets x_1$, $h_{1}\gets y_{1}$.
      \FOR{$t = 1, \dots, T$} \STATE Get $T$ independent samples
      ${i_1^t}, {i_2^t},\dots, {i_{T}^t}$ from the Gibbs sampling oracle
      $\Oracle_{-\mathbf{A}h_{t}}\rbra*{T,\epsilon_{\G}}$.
      \STATE Choose the action $\zeta_{t} = \sum_{N = 1}^T e_{i_N^t}/T$.
      \STATE Update $x_{t+1} \gets x_t + \lambda \zeta_{t}$,
       $g_{t+1} \gets x_{t+1} + \lambda \zeta_{t}$.
       $\hat{u}\gets \hat{u}+\frac{1}{T}\zeta_{t}$.
      \STATE Get $T$ independent samples $j_1^t, j_2^t,\dots, j_{T}^t$ from the
      Gibbs sampling oracle
      $\Oracle_{\mathbf{A}^\trans g_{t}}\rbra*{T,\epsilon_{\G}}$.
      \STATE Choose the action $\eta_{t} = \sum_{N = 1}^T e_{j^t_N}/T$. 
      \STATE Update $y_{t+1} \gets y_t + \lambda \eta_t$,
      $h_{t+1} \gets y_{t+1} + \lambda \eta_t$,
      $\hat{v}\gets \hat{v}+\frac{1}{T}\eta_t$.
      \ENDFOR
      \STATE \textbf{return} $\left(\hat{u},\hat{v}\right)$.
    \end{algorithmic}
  \end{algorithm}
  
  Suppose the zero-sum game is represented by matrix
  $\mbf{A} \in \mr^{m \times n}$ with $\Abs{\mbf{A}} \leq 1$.
  Our sample-based optimistic multiplicative weight update algorithm is given in
  \cref{algo:POMWU}.
  \Cref{algo:POMWU} is inspired by the classical optimistic
  follow-the-regularized-leader algorithm (see \cref{append:opt-online-learning}
  for more information).
  In that classical algorithm, the update terms are essentially $\E\sbra{\zeta_t}$
  and $\E\sbra{\eta_t}$, which are computed deterministically by matrix arithmetic
  operations during the update.
  In contrast, we do this probabilistically by sampling from the Gibbs
  distributions and $\zeta_t$ and $\eta_t$ in Line 5 and Line 10 are the
  corresponding averages of the samples.
  For this to work, we need to bound the expectation of the total regret (see
  \cref{append:algoanalysis}) based on the RVU property (\cref{def:RVU}).
  Technically, the $l_1$ variances of $\zeta_t$ and $\eta_t$ turn out to be
  significant in the analysis.
  To reduce the variances, we need multiple independent samples identically
  distributed from Gibbs distributions (see Line 4 and Line 7 in
  \cref{algo:POMWU}).
  Because of the randomness from Gibbs sampling oracles, the total regret
  $\mcal{R}\rbra*{T} + \mcal{R}'\rbra*{T}$ of \cref{algo:POMWU} is a random
  variable.
  Nevertheless, we can bound the total regret by $O\rbra{\log\rbra{mn}}$ with high
  probability as follows.
  
  \begin{thm}\label{thm:error-Gibbs}
    Let $\epsilon_{\G} = 1/T$.
    After $T$ rounds of playing the zero-sum game $\mbf{A}$ by~\cref{algo:POMWU},
    the total regret will be bounded by
    \begin{equation*}
      \mcal{R}\rbra*{T} + \mcal{R}'\rbra*{T}\le 144 \lambda +
      \frac{3\log \rbra*{mn}}{\lambda}+ 12,
    \end{equation*}
    with probability at least $2/3$.
    Then, for any constant $\lambda \in \rbra{0, \sqrt{3}/6}$, the total regret is
    $O\rbra{\log\rbra{mn}}$.
    Moreover, if we choose $T = \widetilde \Theta\rbra{1/\eps}$, then
    \cref{algo:POMWU} will return an $\eps$-approximate Nash equilibrium of the
    game $\mbf{A}$.
  \end{thm}
  
  The proof of \cref{thm:error-Gibbs} is deferred to \cref{append:sample-with-error}.
  Combining \cref{thm:error-Gibbs} and our fast quantum multi-Gibbs sampler in
  \cref{thm:gibbs} (which will be developed in the next section), we obtain a
  quantum algorithm for finding an $\eps$-approximate Nash equilibrium of zero-sum
  games.
  
  \begin{cor}\label{cor:main}
    If we choose $T = \widetilde \Theta\rbra{1/\eps}$
    and use quantum multi-Gibbs sampling (\cref{algo:gibbs}) with
    $\epsilon_{\G} = 1/T$, then \cref{algo:POMWU} will return an
    $\eps$-approximate Nash equilibrium in quantum time
    $\widetilde O \rbra{\sqrt{m+n}/\eps^{2.5}}$.
  \end{cor}
  
  \section{Fast Quantum Multi-Gibbs Sampling}\label{sec:multi-Gibbs}
  
  In \cref{algo:POMWU}, the vectors $h_t$ and $g_t$ updated in each round are used
  to generate independent samples from Gibbs distributions
  $\Oracle_{-\mathbf{A}h_{t}}$ and $\Oracle_{\mathbf{A}^\intercal g_{t}}$.
  Here, $h_t$ and $g_t$ are supposed to be stored in classical registers.
  To allow quantum speedups, we store $h_t$ and $g_t$ in
  quantum-read classical-write random access memory (QRAM)~\cite{GLM2008}, which is commonly used in prior work~\cite{vG2019,BGJST2022}.
  Specifically, a QRAM can store/modify an array $a_1, a_2, \dots, a_{\ell}$ of
  classical data and provide quantum (read-only) access to them, i.e., a unitary
  operator $U_{\text{QRAM}}$ is given such that
  \[
  U_{\text{QRAM}} \colon \ket{i} \ket{0} \mapsto \ket{i} \ket{a_i}.
  \]
  Without loss of generality (see \cref{rem:wlog}), suppose we have quantum oracle
  $\mathcal{O}_{\mbf{A}}$ for $\mbf{A} \in \mr^{n \times n}$ with
  $A_{i,j} \in \sbra*{0, 1}$, and QRAM access to a vector $z \in \mr^{n}$ with
  $z_i\ge 0$. 
  
  We also need the polynomial approximation of the exponential function for applying the QSVT technique~\cite{GSLW2019}:
  \begin{lmm} [Polynomial approximation, Lemma 7
    of~\cite{vAG2019}]\label{lem:poly-approx} Let $\beta \geq 1$ and
    $\epsilon_P \in \rbra*{0, 1/2}$.
    There is a classically efficiently computable polynomial
    $P_\beta \in \mr\sbra*{x}$ of degree
    $O\rbra{\beta \log\rbra{\epsilon_P^{-1}}}$ such that
    $\abs{P_\beta\rbra*{x}} \leq 1$ for $x \in \sbra*{-1, 1}$, and
    \[
      \max_{x\in \sbra*{-1, 0}} \abs{P_\beta\rbra*{x} - \frac{1}{4}\exp\rbra*{\beta x}} \leq \epsilon_P.
    \]
  \end{lmm}
  Then, we can produce multiple samples from $\Oracle_{\mbf{A}z}$ efficiently by
  \cref{algo:gibbs} on a quantum computer. \cref{algo:gibbs} is inspired by~\cite{H2022} about preparing multiple samples
  of a quantum state, with quantum access to its amplitudes.
  However, we do not have access to the exact values of the amplitudes, which are
  $\rbra{\mbf{A}z}_i$ in our case.
  To resolve this issue, we develop consistent quantum amplitude estimation (see
  \cref{append:consistent-amplitude-estimation}) to estimate $\rbra{\mbf{A}z}_i$
  with a unique answer (Line 1).
  After having prepared an initial quantum state $\ket{u_{\text{guess}}}$, we
  use quantum singular value decomposition~\cite{GSLW2019} to correct the tail
  amplitudes (Line 6), and finally obtain the desired quantum state
  $\ket{\tilde u_{\text{Gibbs}}}$ by quantum amplitude amplification~\cite{BHMT02}
  (Line 7). We have the following (see \cref{append:multi-Gibbs-proof} for its proof): 
  
  \begin{algorithm}[t]
    \caption{Quantum Multi-Gibbs Sampling
      $\Oracle_{\mbf{A}z}\rbra{k,\epsilon_{\G}}$}\label{algo:gibbs}
    \begin{algorithmic}[1]
      \REQUIRE Quantum oracle $\mathcal{O}_{\mbf{A}}$ for
      $\mathbf{A}\in \mr^{n \times n}$, QRAM access to $z \in \mr^n$ with
      $\Abs{z}_1 \leq \beta$, polynomial $P_{2\beta}$ with parameters $\epsilon_P$
      by \cref{lem:poly-approx}, number $k$ of samples.
      We write $u = \mbf{A} z$.
      \ENSURE $k$ independent samples $i_1, i_2, \dots, i_k$.
      \STATE Obtain
      $\mathcal{O}_{\tilde u} \colon \ket{i}\ket{0} \mapsto \ket{i}\ket{\tilde u_i}$
      by consistent quantum amplitude estimation such that
      $u_i \leq \tilde u_i \leq u_i + 1$.
      \STATE Find the set $S \subseteq \sbra*{n}$ of indexes of the $k$ largest
      $\tilde u_i$ by quantum $k$-maximum finding, with access to
      $\mathcal{O}_{\tilde u}$.
      \STATE Obtain $\tilde u_i$ for all $i \in S$ from $\mathcal{O}_{\tilde u}$,
      then compute $\tilde u^* = \min\limits_{i \in S} \tilde u_i$ and
      $W = \sum\limits_{i \in S} \exp\rbra*{\tilde u_i} + \rbra*{n-k} \exp\rbra*{\tilde u^*}$.
      \FOR{$\ell = 1, \dots, k$}
      \STATE Prepare the quantum state
      $
        \ket{u_{\text{guess}}} = \sum_{i \in S} \sqrt{\frac{\exp\rbra*{\tilde u_i}}{W} } \ket{i} +
        \sum_{i \notin S} \sqrt{\frac{\exp\rbra*{\tilde u^*}}{W} } \ket{i}.
      $
      \STATE Obtain unitary $U^{\exp}$ such that
      $\bra{0}^{\otimes a} U^{\exp} \ket{0}^{\otimes a} =
      \diag\rbra*{P_{2\beta}\rbra{u - \max\cbra*{\tilde u, \tilde u^*}}} / 4\beta$
      by QSVT.
      \STATE Post-select
      $\ket{\tilde u_{\text{Gibbs}}} \propto \bra{0}^{\otimes a} U^{\exp}
      \ket{u_{\text{guess}}} \ket{0}^{\otimes a}$
      by quantum amplitude amplification.
      \STATE Let $i_\ell$ be the measurement outcome of
      $\ket{\tilde u_{\text{Gibbs}}}$ in the computational basis.
      \ENDFOR
      \STATE \textbf{return} $i_1, i_2, \dots, i_k$.
    \end{algorithmic}
  \end{algorithm}
  
  \begin{thm} [Fast quantum multi-Gibbs sampling]\label{thm:gibbs}
    For $k\in [n]$, if we set
    $\epsilon_P = \Theta \rbra{k\epsilon^2_{\G}/n}$, then
    \cref{algo:gibbs} will produce $k$ independent and identical distributed
    samples from a distribution that is $\epsilon_{\G}$-close to
    $\Oracle_{\mbf{A}z}$ in total variation distance, in quantum time
    $\widetilde O\rbra{\beta \sqrt{nk}}$.
  \end{thm}
  
  \begin{rmk}\label{rem:wlog} If $\mbf{A} \in \mr^{m \times n}$ is not a square
    matrix, then by adding $0$'s we can always enlarge $\mbf{A}$ to an
    $\rbra*{m+n}$-dimensional square matrix.
    For $\Oracle_{-\mbf{A}h_t}$ as required in \cref{algo:POMWU}, we note that
    $\Oracle_{\rbra*{\mathbf{1}-\mbf{A}}h_t}$ indicates the same distribution as
    $\Oracle_{-\mbf{A}h_t}$, where $\mathbf{1}$ has the same size as $\mbf{A}$
    with all entries being $1$.
    From the above discussion, we can always convert $\mbf{A}$ to another matrix
    satisfying our assumption, i.e., with entries in the range $\sbra*{0, 1}$.
  \end{rmk}
  
  \begin{rmk}\label{rem:term}
    The description of the unitary operator $U^{\exp}$ defined by the polynomial
    $P_{2\beta}$ can be classically computed to
    precision $\epsilon_P$ in time $\widetilde O\rbra{\beta^3}$ by~\cite{H2019},
    which is $\widetilde O\rbra{1/\eps^3}$ in our case as
    $\beta \leq \lambda T = \widetilde \Theta\rbra{1/\eps}$ is required in
    \cref{cor:main}.
    This extra cost can be neglected because
    $\widetilde O\rbra{\sqrt{m+n}/\eps^{2.5}}$ dominates the complexity whenever
    $\eps = \Omega\rbra{\rbra{m+n}^{-1}}$.
  \end{rmk}

  \begin{rmk}
    Our multi-Gibbs sampler is based on maximum finding and consistent amplitude estimation, with a guaranteed worst-case performance in each round. 
    In comparison, the dynamic Gibbs sampler in \cite{BGJST2022} maintains a hint vector, resulting in an amortized time complexity per round.
  \end{rmk}
  
  \section{Discussion}
  
  In our paper, we propose the first quantum online algorithm for zero-sum games
  with near-optimal regret.
  This is achieved by developing a sample-based stochastic version of the
  optimistic multiplicative weight update method~\cite{SALS2015}.
  Our core technical contribution is a fast multi-Gibbs sampling, which may have
  potential applications in other quantum computing scenarios.
  
  Our result naturally gives rise to some further open questions.
  For instance: Can we improve the dependence on $\eps$ for the time complexity?
  And can we further explore the combination of optimistic learning and quantum
  computing into broader applications?
  Now that many heuristic quantum approaches for different machine learning
  problems have been realized, e.g.in~\cite{HCTHKCG2019,SAHS2021,HSNS2021}, can
  fast quantum algorithms for zero-sum games be realized in the near future?

  \section*{Acknowledgments}
  
  We would like to thank Kean Chen, Wang Fang, Ji Guan, Junyi Liu, Xinzhao Wang,
  Chenyi Zhang, and Zhicheng Zhang for helpful discussions.
  Minbo Gao would like to thank Mingsheng Ying for valuable suggestions.
  
  Minbo Gao was supported by the National Key R\&D Program of China (2018YFA0306701) and
  the National Natural Science Foundation of China (61832015).
  Zhengfeng Ji was supported by a startup fund from Tsinghua University, and the Department
  of Computer Science and Technology, Tsinghua Univeristy.
  Tongyang Li was supported by a startup fund from Peking University, and the Advanced
  Institute of Information Technology, Peking University.
  Qisheng Wang was supported by the MEXT Quantum Leap Flagship Program (MEXT Q-LEAP) grants
  No.\ JPMXS0120319794.

\bibliographystyle{alpha}
\bibliography{zero-sum.bib}

\newcommand{\etalchar}[1]{$^{#1}$}
\begin{thebibliography}{vAGGdW17}

\bibitem[ACH{\etalchar{+}}18]{ACHKN2018}
Scott Aaronson, Xinyi Chen, Elad Hazan, Satyen Kale, and Ashwin Nayak.
\newblock Online learning of quantum states.
\newblock In {\em Advances in Neural Information Processing Systems}, volume~31, 2018.

\bibitem[ADF{\etalchar{+}}22]{ADFFGS2022}
Ioannis Anagnostides, Constantinos Daskalakis, Gabriele Farina, Maxwell Fishelson, Noah Golowich, and Tuomas Sandholm.
\newblock Near-optimal no-regret learning for correlated equilibria in multi-player general-sum games.
\newblock In {\em Proceedings of the 54th Annual ACM SIGACT Symposium on Theory of Computing}, pages 736--749, 2022.

\bibitem[AHK12]{AHK2012}
Sanjeev Arora, Elad Hazan, and Satyen Kale.
\newblock The multiplicative weights update method: a meta-algorithm and applications.
\newblock {\em Theory of Computing}, 8(6):121--164, 2012.

\bibitem[Amb12]{A2012}
Andris Ambainis.
\newblock Variable time amplitude amplification and quantum algorithms for linear algebra problems.
\newblock In {\em Proceedings of the 29th Symposium on Theoretical Aspects of Computer Science}, volume~14, pages 636--647, 2012.

\bibitem[BCC{\etalchar{+}}15]{BCCKS2015}
Dominic~W. Berry, Andrew~M. Childs, Richard Cleve, Robin Kothari, and Rolando~D. Somma.
\newblock Simulating {Hamiltonian} dynamics with a truncated {Taylor} series.
\newblock {\em Physical Review Letters}, 114(9):090502, 2015.

\bibitem[BGJ{\etalchar{+}}23]{BGJST2022}
Adam Bouland, Yosheb~M Getachew, Yujia Jin, Aaron Sidford, and Kevin Tian.
\newblock Quantum speedups for zero-sum games via improved dynamic {G}ibbs sampling.
\newblock In {\em Proceedings of the 40th International Conference on Machine Learning}, volume 202, pages 2932--2952, 2023.

\bibitem[BHMT02]{BHMT02}
Gilles Brassard, Peter H{\o}yer, Michele Mosca, and Alain Tapp.
\newblock Quantum amplitude amplification and estimation.
\newblock {\em Quantum Computation and Information}, 305:53--74, 2002.

\bibitem[BKL{\etalchar{+}}19]{BKLLSW2019}
Fernando G. S.~L. Brand{\~{a}}o, Amir Kalev, Tongyang Li, Cedric Yen-Yu Lin, Krysta~M. Svore, and Xiaodi Wu.
\newblock Quantum {SDP} solvers: Large speed-ups, optimality, and applications to quantum learning.
\newblock In {\em Proceedings of the 46th International Colloquium on Automata, Languages, and Programming}, volume 132, pages 27:1--27:14, 2019.

\bibitem[BS17]{BS2017}
Fernando~G.S.L. Brand{\~{a}}o and Krysta~M. Svore.
\newblock Quantum speed-ups for solving semidefinite programs.
\newblock In {\em Proceedings of the 58th Annual Symposium on Foundations of Computer Science}, pages 415--426, 2017.

\bibitem[CBL06]{CNL2006}
Nicolo Cesa-Bianchi and Gabor Lugosi.
\newblock {\em Prediction, Learning, and Games}.
\newblock Cambridge University Press, 2006.

\bibitem[CD06]{CD2006}
Xi~Chen and Xiaotie Deng.
\newblock Settling the complexity of two-player nash equilibrium.
\newblock In {\em Proceedings of the 47th Annual Symposium on Foundations of Computer Science}, pages 261--272, 2006.

\bibitem[CHL{\etalchar{+}}22]{CHLLWY2022}
Xinyi Chen, Elad Hazan, Tongyang Li, Zhou Lu, Xinzhao Wang, and Rui Yang.
\newblock Adaptive online learning of quantum states, 2022.

\bibitem[CJST19]{CJST2019}
Yair Carmon, Yujia Jin, Aaron Sidford, and Kevin Tian.
\newblock Variance reduction for matrix games.
\newblock In {\em Advances in Neural Information Processing Systems}, volume~32, pages 11377--11388, 2019.

\bibitem[CW12]{CW2012}
Andrew~M. Childs and Nathan Wiebe.
\newblock Hamiltonian simulation using linear combinations of unitary operations.
\newblock {\em Quantum Information and Computation}, 12(11--12):901--924, 2012.

\bibitem[CW20]{CW2020}
Yifang Chen and Xin Wang.
\newblock More practical and adaptive algorithms for online quantum state learning, 2020.

\bibitem[DDK15]{DDK2015}
Constantinos Daskalakis, Alan Deckelbaum, and Anthony Kim.
\newblock Near-optimal no-regret algorithms for zero-sum games.
\newblock {\em Games and Economic Behavior}, 92:327--348, 2015.

\bibitem[DGP09]{DGP2009}
Constantinos Daskalakis, Paul~W. Goldberg, and Christos~H. Papadimitriou.
\newblock The complexity of computing a nash equilibrium.
\newblock {\em SIAM Journal on Computing}, 39(1):195--259, 2009.

\bibitem[DHHM06]{DHHM2006}
Christoph D{\"{u}}rr, Mark Heiligman, Peter H{\o}yer, and Mehdi Mhalla.
\newblock Quantum query complexity of some graph problems.
\newblock {\em SIAM Journal on Computing}, 35(6):1310--1328, 2006.

\bibitem[DISZ18]{DISZ2018}
Constantinos Daskalakis, Andrew Ilyas, Vasilis Syrgkanis, and Haoyang Zeng.
\newblock Training gans with optimism.
\newblock In {\em Proceedings of the 6th International Conference on Learning Representations}, 2018.

\bibitem[GK95]{GK1995}
Michael~D. Grigoriadis and Leonid~G. Khachiyan.
\newblock A sublinear-time randomized approximation algorithm for matrix games.
\newblock {\em Operations Research Letters}, 18(2):53--58, 1995.

\bibitem[GLM08]{GLM2008}
Vittorio Giovannetti, Seth Lloyd, and Lorenzo Maccone.
\newblock Quantum random access memory.
\newblock {\em Physical Review Letters}, 100(16):160501, 2008.

\bibitem[GR02]{GR2002}
Lov Grover and Terry Rudolph.
\newblock Creating superpositions that correspond to efficiently integrable probability distributions, 2002.

\bibitem[Gro00]{G2000}
Lov~K. Grover.
\newblock Synthesis of quantum superpositions by quantum computation.
\newblock {\em Physical Review Letters}, 85(6):1334, 2000.

\bibitem[GSLW19]{GSLW2019}
András Gilyén, Yuan Su, Guang~Hao Low, and Nathan Wiebe.
\newblock Quantum singular value transformation and beyond: Exponential improvements for quantum matrix arithmetics.
\newblock In {\em Proceedings of the 51st Annual ACM SIGACT Symposium on Theory of Computing}, pages 193--204, 2019.

\bibitem[Haa19]{H2019}
Jeongwan Haah.
\newblock Product decomposition of periodic functions in quantum signal processing.
\newblock {\em Quantum}, 3:190, 2019.

\bibitem[HAM21]{HAM2021}
Yu-Guan Hsieh, Kimon Antonakopoulos, and Panayotis Mertikopoulos.
\newblock Adaptive learning in continuous games: Optimal regret bounds and convergence to nash equilibrium.
\newblock In {\em Proceedings of the 34th Conference on Learning Theory}, volume 134, pages 2388--2422, 2021.

\bibitem[Ham22]{H2022}
Yassine Hamoudi.
\newblock Preparing many copies of a quantum state in the black-box model.
\newblock {\em Physical Review A}, 105(6):062440, 2022.

\bibitem[HCT{\etalchar{+}}19]{HCTHKCG2019}
Vojtěch Havlíček, Antonio~D. Córcoles, Kristan Temme, Aram~W. Harrow, Abhinav Kandala, Jerry~M. Chow, and Jay~M. Gambetta.
\newblock Supervised learning with quantum-enhanced feature spaces.
\newblock {\em Nature}, 567(7747):209--212, 2019.

\bibitem[HHL09]{HHL2009}
Aram~W. Harrow, Avinatan Hassidim, and Seth Lloyd.
\newblock Quantum algorithm for linear systems of equations.
\newblock {\em Physical Review Letters}, 103(15):150502, 2009.

\bibitem[HSN{\etalchar{+}}21]{HSNS2021}
Matthew~P. Harrigan, Kevin~J. Sung, Matthew Neeley, Kevin~J. Satzinger, Frank Arute, Kunal Arya, Juan Atalaya, Joseph~C. Bardin, Rami Barends, Sergio Boixo, Michael Broughton, Bob~B. Buckley, David~A. Buell, Brian Burkett, Nicholas Bushnell, Yu~Chen, Zijun Chen, {Ben Chiaro}, Roberto Collins, William Courtney, Sean Demura, Andrew Dunsworth, Daniel Eppens, Austin Fowler, Brooks Foxen, Craig Gidney, Marissa Giustina, Rob Graff, Steve Habegger, Alan Ho, Sabrina Hong, Trent Huang, L.~B. Ioffe, Sergei~V. Isakov, Evan Jeffrey, Zhang Jiang, Cody Jones, Dvir Kafri, Kostyantyn Kechedzhi, Julian Kelly, Seon Kim, Paul~V. Klimov, Alexander~N. Korotkov, Fedor Kostritsa, David Landhuis, Pavel Laptev, Mike Lindmark, Martin Leib, Orion Martin, John~M. Martinis, Jarrod~R. McClean, Matt McEwen, Anthony Megrant, Xiao Mi, Masoud Mohseni, Wojciech Mruczkiewicz, Josh Mutus, Ofer Naaman, Charles Neill, Florian Neukart, Murphy~Yuezhen Niu, Thomas~E. O’Brien, Bryan O’Gorman, Eric Ostby, Andre Petukhov, Harald Putterman, Chris Quintana, Pedram Roushan, Nicholas~C. Rubin, Daniel Sank, Andrea Skolik, Vadim Smelyanskiy, Doug Strain, Michael Streif, Marco Szalay, Amit Vainsencher, Theodore White, Z.~Jamie Yao, Ping Yeh, Adam Zalcman, Leo Zhou, Hartmut Neven, Dave Bacon, Erik Lucero, Edward Farhi, and Ryan Babbush.
\newblock Quantum approximate optimization of non-planar graph problems on a planar superconducting processor.
\newblock {\em Nature Physics}, 17(3):332--336, 2021.

\bibitem[KP17]{KP2017}
Iordanis Kerenidis and Anupam Prakash.
\newblock Quantum recommendation systems.
\newblock In {\em Proceedings of the 8th Innovations in Theoretical Computer Science Conference}, volume~67, pages 49:1--49:21, 2017.

\bibitem[KWS16]{KWS2016}
Ashish Kapoor, Nathan Wiebe, and Krysta Svore.
\newblock Quantum perceptron models.
\newblock In {\em Advances in Neural Information Processing Systems}, volume~29, 2016.

\bibitem[LCW19]{LCW2019}
Tongyang Li, Shouvanik Chakrabarti, and Xiaodi Wu.
\newblock Sublinear quantum algorithms for training linear and kernel-based classifiers.
\newblock In {\em Proceedings of the 36th International Conference on Machine Learning}, volume~97, pages 3815--3824, 2019.

\bibitem[LMR14]{LMR2014}
Seth Lloyd, Masoud Mohseni, and Patrick Rebentrost.
\newblock Quantum principal component analysis.
\newblock {\em Nature Physics}, 10(9):631--633, 2014.

\bibitem[Mye81]{M1981}
Roger~B. Myerson.
\newblock Optimal auction design.
\newblock {\em Mathematics of Operations Research}, 6(1):58--73, 1981.

\bibitem[Nem04]{A04}
Arkadi Nemirovski.
\newblock Prox-method with rate of convergence o(1/t) for variational inequalities with lipschitz continuous monotone operators and smooth convex-concave saddle point problems.
\newblock {\em SIAM Journal on Optimization}, 15(1):229--251, 2004.

\bibitem[Nes05a]{Yu05b}
Yu. Nesterov.
\newblock Excessive gap technique in nonsmooth convex minimization.
\newblock {\em SIAM Journal on Optimization}, 16(1):235--249, 2005.

\bibitem[Nes05b]{Yu05a}
Yu. Nesterov.
\newblock Smooth minimization of non-smooth functions.
\newblock {\em Mathematical Programming}, 103(1):127–152, may 2005.

\bibitem[ORR13]{ORR2013}
Maris Ozols, Martin Roetteler, and J{\'{e}}r{\'{e}}mie Roland.
\newblock Quantum rejection sampling.
\newblock {\em ACM Transactions on Computation Theory}, 5(3):11:1--33, 2013.

\bibitem[PW09]{PW2009}
David Poulin and Pawel Wocjan.
\newblock Sampling from the thermal quantum {Gibbs} state and evaluating partition functions with a quantum computer.
\newblock {\em Physical Review Letters}, 103(22):220502, 2009.

\bibitem[PZO22]{PZO2022}
Sarath Pattathil, Kaiqing Zhang, and Asuman Ozdaglar.
\newblock Symmetric (optimistic) natural policy gradient for multi-agent learning with parameter convergence, 2022.

\bibitem[RML14]{RML2014}
Patrick Rebentrost, Masoud Mohseni, and Seth Lloyd.
\newblock Quantum support vector machine for big data classification.
\newblock {\em Physical Review Letters}, 113:130503, Sep 2014.

\bibitem[RT02]{RT2002}
Tim Roughgarden and \'{E}va Tardos.
\newblock How bad is selfish routing?
\newblock {\em Journal of the ACM}, 49(2):236–259, mar 2002.

\bibitem[SAH{\etalchar{+}}21]{SAHS2021}
V.~Saggio, B.~E. Asenbeck, A.~Hamann, T.~Strömberg, P.~Schiansky, V.~Dunjko, N.~Friis, N.~C. Harris, M.~Hochberg, D.~Englund, S.~Wölk, H.~J. Briegel, and P.~Walther.
\newblock Experimental quantum speed-up in reinforcement learning agents.
\newblock {\em Nature}, 591(7849):229--233, 2021.

\bibitem[SALS15]{SALS2015}
Vasilis Syrgkanis, Alekh Agarwal, Haipeng Luo, and Robert~E. Schapire.
\newblock Fast convergence of regularized learning in games.
\newblock In {\em Advances in Neural Information Processing Systems}, pages 2989--2997, 2015.

\bibitem[Sch80]{S1980}
Thomas~C Schelling.
\newblock {\em The Strategy of Conflict: with a new Preface by the Author}.
\newblock Harvard university press, 1980.

\bibitem[TS13]{T2013}
Amnon Ta-Shma.
\newblock Inverting well conditioned matrices in quantum logspace.
\newblock In {\em Proceedings of the 45th Annual ACM Symposium on Theory of Computing}, pages 881--890, 2013.

\bibitem[vAG19a]{vG2019}
Joran van Apeldoorn and András Gilyén.
\newblock Improvements in quantum {SDP}-solving with applications.
\newblock In {\em Proceedings of the 46th International Colloquium on Automata, Languages, and Programming}, volume 132, pages 99:1--99:15, 2019.

\bibitem[vAG19b]{vAG2019}
Joran van Apeldoorn and András Gilyén.
\newblock Quantum algorithms for zero-sum games, 2019.

\bibitem[vAGGdW17]{vGGd2017}
Joran van Apeldoorn, Andr{\'{a}}s Gily{\'{e}}n, Sander Gribling, and Ronald de~Wolf.
\newblock Quantum {SDP}-solvers: Better upper and lower bounds.
\newblock In {\em Proceedings of the 58th Annual Symposium on Foundations of Computer Science}, pages 403--414, 2017.

\bibitem[vN28]{vN28}
John von Neumann.
\newblock Zur theorie der gesellschaftsspiele.
\newblock {\em Mathematische Annalen}, 100(1):295--320, 1928.

\bibitem[YJZS20]{YJZS2020}
Feidiao Yang, Jiaqing Jiang, Jialin Zhang, and Xiaoming Sun.
\newblock Revisiting online quantum state learning.
\newblock {\em Proceedings of the AAAI Conference on Artificial Intelligence}, 34(04):6607--6614, 2020.

\end{thebibliography}

\newpage
\appendix

\onecolumn

\section{Revisit of Optimistic Online Learning}\label{append:opt-online-learning}
  
In this section, we briefly review some important properties in the classical
optimistic online learning algorithms.
Some of the propositions in this section will be frequently used in the proof of
the regret bound.

For convenience, we will use $\psi\rbra*{\cdot}$ to denote the negative entropy
function, i.e., $\psi \colon \Delta_n \to \mr$,
$\psi \rbra*{p} = \sum_{i=1}^n p_i \log p_i$.
Note that $\log$ stands for the natural logarithm function with base
$\mathrm{e}$.

For a vector norm $\norm{\cdot}$, its dual norm is defined as:
\begin{equation*}
 \norm{y}_{*} = \max_{x} \set{\Ip{x}{y} : \norm{x}\le 1}.
\end{equation*}

\begin{prop}\label{prop:Variational-Gibbs}
  Let $L$ be a vector in $n$-dimensional space.
  If $p^{*} = \argmin_{p\in \Delta_n } \set{ \langle p, L\rangle + \psi\left(p\right)}$,
  then $p^*$ can be written as:
  \begin{equation*}
    p^* = \GibbsDis{-L}
  \end{equation*}
  and vice versa.
\end{prop}

\begin{proof}
  Write $L = \rbra*{L_1,L_2,\ldots, L_n}$.
  By definition, we know that $p^*$ is the solution to the following convex
  optimization problem:
  \begin{mini*} {p_1, p_2,\ldots, p_n}{\sum_{i=1}^n p_i L_i + \sum_{i=1}^n p_i
     \log p_i} {\label{eq:Optimize-Gibbs}}{}
      \addConstraint{\sum_{i=1}^n p_i}{=1}
      \addConstraint{\forall i, p_i}{\ge 0}
  \end{mini*}
  The Lagrangian is
  \begin{equation*}
    \mathcal{L}(p,u,v) =   \sum_{i=1}^n p_i L_i + \sum_{i=1}^n p_i \log p_i -
    \sum_{i=1}^n u_i p_i + v\rbra*{\sum_{i=1}^n p_i -1}
  \end{equation*}
  From KKT conditions, we know that the stationarity is:
  \begin{equation}\label{eq:KKTSta}
    L_i + 1 + \log p_i - u_i +v = 0.
  \end{equation}
  The complementary slackness is:
  \begin{equation*}
    u_i p_i = 0.
  \end{equation*}
  The primal feasibility is
  \begin{equation*}
    \forall i,  p_i \ge 0; \sum_{i=1}^n p_i = 1.
  \end{equation*}
  The dual feasibility is
  \begin{equation*}
    u_i\ge 0.
  \end{equation*}
  If $u_i\ne 0$ then $p_i =0$, from stationarity we know
  $u_i = -\infty$, but that violates the dual feasibility.
  So we can conclude that $u_i = 0$ for all $i\in [n]$, thus
  $p_i \propto \exp \rbra*{-L_i}$ and the result follows.
\end{proof}

Now we present a generalized version of the optimistic multiplicative weight
algorithm called optimisitically follow the regularized leader (Opt-FTRL) in
\cref{algo:OptFTRL}.
In the algorithm, $m_t$ has the same meaning as $m\rbra*{t}$ for notation
consistency.

\begin{algorithm}[H]
    \caption{Optimistic follow-the-regularized-leader}\label{algo:OptFTRL}
    \begin{algorithmic}
      \REQUIRE The closed convex domain $\mcal{X}$.
      \ENSURE Step size $\lambda$, loss gradient prediction $m$.
      \STATE Initialize $L_0 \leftarrow 0$, choose appropriate $m_1$.
      \FOR{$t = 1, \ldots, T$}
      \STATE \textbf{Choose}
      $x_t = \argmin_{x\in \mcal{X}} \set{\lambda\Ip{L_{t-1}+m_t}{x}+\psi\rbra*{x}}$.
      \STATE Observe loss $l_t$, update $L_t = L_{t-1}+l_t$.
      \STATE Compute $m_{t+1}$ using observations till now.
      \ENDFOR
    \end{algorithmic}
\end{algorithm}

Now we study a crucial property that leads to the fast convergence of the
algorithm, called the Regret bounded by Variation in Utilities (RVU in short).
For simplicity, we only consider the linear loss function $l_t(x) = \Ip{l_t}{x}$.
(There is a little abuse of notation here.)

\begin{defi}[Regret bounded by Variations in Utilities (RVU),
  Definition $3$ in~\cite{SALS2015}]\label{def:RVU}
  Consider an online learning algorithm $\mcal{A}$ with regret
  $\mcal{R}\rbra*{T} = o\rbra*{T}$, we say that it has the property of regret
  bound by variation in utilities if for any linear loss sequence
  $l_1, l_2,\ldots, l_T$, there exists parameters
  $\alpha > 0, 0 < \beta \le \gamma$ such that the algorithm output decisions
  $x_1,x_2,\ldots, x_T, x_{T+1}$ that satisfy:
  \begin{equation*}
    \sum_{i=1}^T \Ip{l_i}{x_i} - \min_{x\in \mcal{X}}
    \sum_{i=1}^T \Ip{l_i}{x} \le \alpha + \beta \sum_{i=1}^{T-1}
    \norm{l_{i+1}-l_{i}}_{*}^2 - \gamma \sum_{i=1}^{T-1} \norm{x_{i+1}-x_i}^2,
 \end{equation*}
 where $\norm{\cdot}_{*}$ is the dual norm of $\norm{\cdot}$.
\end{defi}

We do not choose the norm to be any specific one here.
In fact,~\cite{SALS2015} have already shown that the above
optimistic follow-the-regularized-leader algorithm has the RVU property with
respect to any norm $\norm{\cdot}$ in which the negative entropy function $\psi$
is $1$-strongly convex.
So, from Pinsker's inequality, for $l_2$ norms the following result holds:

\begin{prop}[Proposition $7$ in~\cite{SALS2015}]\label{prop:OptFTRLRVU}
  If we choose $m_t = l_{t-1}$ in the optimistic follow-the-regularized-leader
  algorithm with step size $\lambda \le 1/2$, then it has the regret bound by
  variation in utilities property with the parameters
  $\alpha = \log n / {\lambda}$, $\beta = \lambda$ and
  $\gamma = 1/\rbra*{4\lambda}$, where $n$ is the dimension of $\mcal{X}$.
\end{prop}

\section{Regret Bound and Time Complexity of Our Algorithm}\label{append:algoanalysis}

\subsection{Ideal Samplers}\label{append:ideal-sampler}

We assume that after the execution of our algorithm, the sequences we get are
$\left\{\rbra*{x_t, y_t}\right\}_{t=1}^{T+1}$ and
$\left\{\rbra*{g_t, h_t}\right\}_{t=1}^{T+1}$, respectively.
We denote $u_t:= \GibbsDis{-\mathbf{A}h_t}$ and
$v_t:=\GibbsDis{\mathbf{A}^\intercal g_t}$ to be the corresponding Gibbs
distribution, we will first assume that the Gibbs oracle in our algorithm has no
error (i.e.
$\epsilon_{\G} = 0$) until \cref{thm:main-theorem-regret-restate} is
proved.

\begin{obs}
  The sequence $\set{u_t}_{t=1}^{T}$ can be seen as the decision result of
  applying optimistic FTRL algorithm to the linear loss function $\mbf{A}\eta_t$
  with linear prediction function $\mbf{A}\eta_{t-1}$, and similarly for
  $\set{v_t}_{t = 1}^{T+1}$ with the loss function
  $-\mathbf{A}^{\trans}\zeta_t$, the prediction function
  $-\mathbf{A}^{\trans}\zeta_{t-1}$.
\end{obs}

\begin{proof}
  By symmetry, we only consider $u_t$.
  Since $u_t = \GibbsDis{-\mbf{A}h_t}$, from \cref{prop:Variational-Gibbs} we
  can write
  \begin{equation*}
    u_t = \argmin_{u\in \Delta_m} \set{\Ip{\mbf{A}h_t}{u} + \psi\rbra*{u}}.
  \end{equation*}
  Then we notice the iteration of \cref{algo:POMWU} gives
  \begin{equation*}
    h_t = \lambda \rbra*{\sum_{i = 1}^{t-1} \eta_i } +\lambda \eta_{t-1}.
  \end{equation*}
  So from the definition of the \cref{algo:OptFTRL}, we know that our observation
  holds.
\end{proof}

This observation, together with \cref{prop:OptFTRLRVU}, gives the following
inequalities.
For any $u\in \Delta_m$, $v \in \Delta_n$, we have:
\begin{align}
  \sum_{t=1}^T \left\langle u_t-u, \mathbf{A}\eta_{t} \right\rangle
  & \leq \frac{\log m}{\lambda}+\lambda \sum_{t=1}^{T-1}
    \left\|\mathbf{A}\rbra*{\eta_{t+1}-\eta_{t}}\right\|^2-
    \frac{1}{4 \lambda} \sum_{t=1}^{T-1}\left\|u_{t+1}-u_t\right\|^2,
   \label{eqn:RVU-u}\\
  \sum_{t=1}^T \left\langle v_t-v, -\mathbf{A}^\intercal \zeta_{t} \right\rangle
  & \leq \frac{\log n}{\lambda}+\lambda \sum_{t=1}^{T-1}\left\|
    \mathbf{A}^\intercal\rbra*{\zeta_{t+1}-\zeta_{t}}\right\|^2
    -\frac{1}{4 \lambda}\sum_{t=1}^{T-1} \left\|v_{t+1}-v_t\right\|^2.
   \label{eqn:RVU-v}
\end{align}

However, we find that the loss function is slightly different from what we expect.

Let us consider the difference $q_t:= \mathbf{A}\rbra*{v_t-\eta_{t}}$ and
$p_t:= -\mathbf{A}^{\trans}\rbra*{u_t-\zeta_{t}}$, we have the decomposition of
the regret:
\begin{equation*}
  \sum_{t = 1}^T \Ip{u_t-u}{\mathbf{A}v_t} =
  \sum_{t = 1}^T \Ip{u_t-u}{\mathbf{A}\eta_t}  + \sum_{t = 1}^T \Ip{u_t-u}{q_t}.
\end{equation*}
Notice that $\E [q_t] = \E [p_t] = 0$, we have:
\begin{lmm}\label{lem:zero_residue}
  \begin{equation*}
    \E \left[\sum_{t = 1}^T \Ip{u_t-u}{q_t}\right] = 0,
    \E \left[\sum_{t = 1}^T \Ip{v_t-v}{p_t}\right] = 0
  \end{equation*}
\end{lmm}

\begin{proof}
  By symmetry, we only prove the case for $u$.
  It suffices to prove that for every $t$,
  $\E \left[ \Ip{u_t-u}{q_t}\right] = 0$.
  Since $u$ is fixed,
  $\E \left[ \Ip{u}{q_t}\right] = \Ip{u}{\E \left[q_t\right]} = 0$.

  Now consider $\E \left[ \Ip{u_t}{q_t}\right]$, notice that given
  $\eta_1, \ldots, \eta_{t-1}$ then $u_t$ is a constant.
  We have:
  \begin{align*}
    \E \left[ \Ip{u_t}{q_t}\right]
    &= \E \left[ \E\left[\Ip{u_t}{q_t}| \eta_1, \eta_2, \ldots, \eta_{t-1} \right]\right] \\
    &= \E \left[ \Ip{u_t}{\E\left[{q_t}| \eta_1, \eta_2, \ldots, \eta_{t-1} \right]}\right]\\
    &= \E \left[ \Ip{u_t}{0}\right] = 0.
 \end{align*}
\end{proof}

Now we are going to bound the term
$\sum_{t=1}^{T-1}\left\|\mathbf{A}\rbra*{\eta_{t+1}-\eta_{t}}\right\|^2 $.

\begin{lmm}
  \begin{align}
    \sum_{t=1}^{T-1}\left\|\mathbf{A}\rbra*{\eta_{t+1}-\eta_{t}}\right\|^2
    &\le 6 + 3 \sum_{t=1}^{T-1} \norm{v_{t+1}-v_{t}}^2,\\
    \sum_{t=1}^{T-1}\left\|\mathbf{A}^{\trans}\rbra*{\zeta_{t+1}-\zeta_{t}}\right\|^2
    &\le 6 + 3 \sum_{t=1}^{T-1} \norm{u_{t+1}-u_{t}}^2.
  \end{align}
\end{lmm}

\begin{proof}

  Recall that by rescaling we have $\norm{\mbf{A}}\le 1$.
  Hence,
  \begin{equation*}
    \sum_{t=1}^{T-1}\left\|\mathbf{A}\rbra*{\eta_{t+1}-\eta_{t}}\right\|^2
    \le  \sum_{t=1}^{T-1}\left\|\eta_{t+1}-\eta_{t}\right\|^2.
  \end{equation*}
  Write
  $\eta_{t+1}-\eta_{t} = \rbra*{\eta_{t+1}-v_{t+1}}+\rbra*{v_{t+1}-v_t}+\rbra*{v_t-\eta_t}$.
  Using the triangle inequality of the $l_1$ norm and the Cauchy inequality
  $\rbra*{a+b+c}^2\le 3\rbra*{a^2+b^2+c^2}$, we get

  \begin{equation}\label{eqn:varbound-v}
    \sum_{t=1}^{T-1}\left\|{\eta_{t+1}-\eta_{t}}\right\|^2  \le
    6 \sum_{t=1}^{T}\left\|\eta_{t}-v_{t}\right\|^2 + 3 \sum_{t=1}^{T-1} \norm{v_{t+1}-v_{t}}^2.
  \end{equation}
    
  Similarly, we have:
  \begin{equation}\label{eqn:varbound-u}
    \sum_{t=1}^{T-1}\left\|{\zeta_{t+1}-\zeta_{t}}\right\|^2  \le
    6 \sum_{t=1}^{T}\left\|\zeta_{t}-u_{t}\right\|^2  + 3 \sum_{t=1}^{T-1} \norm{u_{t+1}-u_{t}}^2.
  \end{equation}
    
  Observing that in our algorithm
  we collect $T$ independent and identically distributed samples and take their
  average, we have:
  \begin{align*}
    \E \sbra*{\sum_{t=1}^{T}\left\|\zeta_{t}-u_{t}\right\|^2} &\le 1,\\
    \E \sbra*{\sum_{t=1}^{T}\left\|\eta_{t}-v_{t}\right\|^2} &\le 1.
  \end{align*}
  Combining the result above, we just get the desired equation.
\end{proof}

We also need the following lemma to guarantee that the sum of the regret is
always non-negative.

\begin{lmm}\label{lem:pos_reg}
  The sum of the regrets of two players in \cref{algo:POMWU} is always
  non-negative.
  In other words:
  \begin{equation*}
    \max_{u\in \Delta_m} \max_{v\in \Delta_n}  \rbra*{\sum_{t=1}^T
      \left\langle u_t-u, \mathbf{A}v_{t} \right\rangle  +
      \sum_{t=1}^T \left\langle v_t-v, -\mathbf{A}^\intercal u_{t} \right\rangle}\ge 0.
  \end{equation*}
\end{lmm}

\begin{proof}
  \begin{align*}
    &\max_{u\in \Delta_m} \max_{v\in \Delta_n}  \rbra*{\sum_{t=1}^T \left\langle u_t-u,
      \mathbf{A}v_{t} \right\rangle  + \sum_{t=1}^T \left\langle
      v_t-v, -\mathbf{A}^\intercal u_{t} \right\rangle} \\
    =& \max_{u\in \Delta_m} \max_{v\in \Delta_n} \rbra*{\sum_{t=1}^T \left\langle -u,
       \mathbf{A}v_{t} \right\rangle  + \sum_{t=1}^T \left\langle v,
       \mathbf{A}^\intercal u_{t} \right\rangle} \\
    =& \max_{v\in \Delta_n} \sum_{t=1}^T \left\langle v, \mathbf{A}^\intercal u_{t} \right\rangle
       -\min_{u\in \Delta_m}  \sum_{t=1}^T \left\langle u, \mathbf{A}v_{t} \right\rangle \ge 0
  \end{align*}

  The last step is because
  \begin{equation*}
    \max_{v\in \Delta_n} \sum_{t=1}^T \left\langle v, \mathbf{A}^\intercal u_{t} \right\rangle
    \ge \Ip{\mbf{A}\sum_{t=1}^T v_t /T }{\sum_{t=1}^T u_t},
  \end{equation*}
  and
  \begin{equation*}
    \min_{u\in \Delta_m}\sum_{t=1}^T \left\langle u, \mathbf{A}v_{t} \right\rangle
    \le \Ip{\mbf{A}\sum_{t=1}^T v_t }{\sum_{t=1}^T u_t/T}.
  \end{equation*}

\end{proof}

Combining the result above, we finally have the following theorem.

\begin{thm}\label{thm:main-theorem-regret-restate} Suppose that in our
  \cref{algo:POMWU}, we choose the episode
  $T = \widetilde \Theta \rbra*{1/\eps}$, and choose a constant learning rate
  $\lambda$ that satisfies $\lambda < \sqrt{3}/6$.
  Then with probability at least $2/3 $ the total regret of the algorithm is
  $\widetilde{O}\rbra*{1}$.
  To be more clear, we have:
  \begin{equation*}
    T\rbra*{\max_{v\in \Delta_n}  \Ip{v}{\mbf{A}^{\trans}\hat{u}} -
      \min_{u\in \Delta_m} \Ip{u}{\mbf{A}\hat{v}} }
    \le 36 \lambda +\frac{3\log \rbra*{mn}}{\lambda},
  \end{equation*}
  and so our algorithm returns an $\eps$-approximate Nash equilibrium.
\end{thm}

\begin{proof}
  Adding the inequalities~\eqref{eqn:RVU-u} and~\eqref{eqn:RVU-v}
  together, we get
  \begin{equation}\label{eqn:reg-first}
    \begin{split}
      \sum_{t=1}^T \left\langle u_t-u, \mathbf{A}\eta_{t} \right\rangle  +
      &\sum_{t=1}^T \left\langle v_t-v, -\mathbf{A}^\intercal \zeta_{t} \right\rangle
        \le  \frac{\log m}{\lambda}  + \frac{\log n}{\lambda} \\
      &+\lambda \sum_{t=1}^{T-1} \left\|\mathbf{A}\rbra*{\eta_{t+1}-\eta_{t}}\right\|^2
        -\frac{1}{4 \lambda}\sum_{t=1}^{T-1} \left\|v_{t+1}-v_t\right\|^2 \\
      &+\lambda \sum_{t=1}^{T-1}\left\|\mathbf{A}^\intercal
        \rbra*{\zeta_{t+1}-\zeta_{t}}\right\|^2-\frac{1}{4 \lambda}
        \sum_{t=1}^{T-1}\left\|u_{t+1}-u_t\right\|^2.
    \end{split}
  \end{equation}

  Taking expectation, and using the inequalities~\eqref{eqn:varbound-v} we have
  \begin{align*}
    &\E \left[\lambda \sum_{t=1}^{T-1} \left\|\mathbf{A}\rbra*{\eta_{t+1}-\eta_{t}}\right\|^2
      -\frac{1}{4 \lambda}\sum_{t=1}^{T-1} \left\|v_{t+1}-v_t\right\|^2\right] \\
    &\le \rbra*{3\lambda-\frac{1}{4\lambda}} \E \left[\sum_{t=1}^{T-1}
      \left\|v_{t+1}-v_t\right\|^2\right] + 6 \lambda \cdot
      \E\left[\sum_{t=1}^{T}\left\|\eta_{t}-v_{t}\right\|^2\right] \\
    & \le 6\lambda.
  \end{align*}
        
  Similarly we can prove
  \begin{equation*}
    \E \left[\lambda \sum_{t=1}^{T-1} \left\|\mathbf{A}^\intercal
          \rbra*{\zeta_{t+1}-\zeta_{t}}\right\|^2
      -\frac{1}{4 \lambda} \sum_{t=1}^{T-1}\left\|u_{t+1}-u_t\right\|^2
    \right]\le 6\lambda.
  \end{equation*}
  So, taking expectations of \cref{eqn:reg-first}, and using the above
  inequalities and the \cref{lem:zero_residue}, we get

  \begin{equation}
    \E \left[\max_{u\in \Delta_m} \sum_{t=1}^T \left\langle u_t-u,
        \mathbf{A}v_{t} \right\rangle  + \max_{v\in \Delta_n} \sum_{t=1}^T
      \left\langle v_t-v, -\mathbf{A}^\intercal u_{t} \right\rangle\right]
    \le 12 \lambda +\frac{\log \rbra*{mn}}{\lambda}.
  \end{equation}

  Using the fact that
  \begin{align*}
    \E \sbra*{\hat{u}}\cdot T &= \sum_{t=1}^T\E \sbra*{u_t}, \\
    \E \sbra*{\hat{v}}\cdot T&= \sum_{t=1}^T\E \sbra*{v_t},
  \end{align*}
  we have
  \begin{equation}
    \E \left[ \max_{v\in \Delta_n}  \Ip{v}{\mbf{A}^{\trans}\hat{u}} -
      \min_{u\in \Delta_m} \Ip{u}{\mbf{A}\hat{v}} \right]
    \cdot T\le 12 \lambda +\frac{\log \rbra*{mn}}{\lambda}.
  \end{equation}
    
  By \cref{lem:pos_reg}, we know that the regret is always non-negative.
  So applying Markov's inequality, we know with probability at least $2/3$, the
  following inequality holds:
  \begin{equation*}
    \max_{v\in \Delta_n} \Ip{v}{\mbf{A}^{\trans}\hat{u}} -
    \min_{u\in \Delta_m} \Ip{u}{\mbf{A}\hat{v}}
    \le \frac{1}{T}\rbra*{36 \lambda +\frac{3\log \rbra*{mn}}{\lambda}}.
 \end{equation*}
\end{proof}

\subsection{Samplers with Errors}\label{append:sample-with-error}

\begin{thm}[Restatement of \cref{thm:error-Gibbs}]\label{thm:main-theorem-error-restate}
  Suppose that in our \cref{algo:POMWU}, we choose the episode
  $T = \widetilde O \rbra*{1/\eps}$, and choose a constant learning rate
  $\lambda$ that satisfies $0 < \lambda < \sqrt{3}/6$.
  The quantum implementation of the oracle in the algorithm will return $T$
  independent and identically distributed samples from a distribution that is
  $\epsilon_{\G}$-close to the desired distribution in total variational
  distance in quantum time $T^Q_{\G}$.

  Then with probability at least $2/3 $ the total regret of the algorithm is
  $\widetilde{O}\rbra*{1+\epsilon_{\G}/\eps}$ and the algorithm
  returns an $\widetilde O \rbra*{\eps+\epsilon_{\G}}$-approximate
  Nash equilibrium in quantum time
  $\widetilde O \rbra{T^Q_{\G}/\eps}$.
\end{thm}

\begin{proof}

  We will follow similar steps of proof for
  \cref{thm:main-theorem-regret-restate}.
  Since the sampling is not from the ideal distribution, we must bound the terms
  where $\eta_t$ and $\zeta_t$ take place.

  Notice that in this case, we have
  \begin{equation*}
    \norm{A \rbra*{v_t - \E \sbra*{\eta_t}}}\le \norm{v_t - \E \sbra*{\eta_t}}\le \epsilon_{\G}.
  \end{equation*}
  So for the term $q_t$ in \cref{lem:zero_residue} we now have the bound:
  \begin{align*}
    &\E \left[\sum_{t = 1}^T \Ip{u_t-u}{A \rbra*{v_t - \eta_t}}\right]\\
    &= \E \left[\sum_{t = 1}^T \Ip{u_t-u}{A \rbra*{v_t - \E\sbra*{\eta_t}}}\right] +
      \E \left[\sum_{t = 1}^T \Ip{u_t-u}{A \rbra*{\E\sbra*{\eta_t}-\eta_t}}\right] \\
    &= \E \left[\sum_{t = 1}^T \Ip{u_t-u}{A \rbra*{v_t - \E\sbra*{\eta_t}}}\right]
      \le 2T \epsilon_{\G},
  \end{align*}
  where the last step is by H\"{o}lder's inequality.

  Then for the other term, we have
  \begin{align*}
    \E \sbra*{\sum_{t=1}^{T}\left\|\eta_{t}-v_{t}\right\|^2}
    &\le 2 \cdot \E \sbra*{\sum_{t=1}^{T}\left\|\eta_{t}-\E \sbra*{\eta_t}\right\|^2} +
      2 \cdot \E \sbra*{\sum_{t=1}^{T}\left\|v_{t}-\E \sbra*{\eta_t}\right\|^2} \\
    & \le 2 + 2 T\epsilon_{\G}^2.
  \end{align*}

  So following the similar steps of proof for
  \cref{thm:main-theorem-regret-restate}, and using the above bounds, we can
  get
  \begin{align*}
    &\E \left[\max_{u\in \Delta_m} \sum_{t=1}^T \left\langle u_t-u,
      \mathbf{A}v_{t} \right\rangle +
      \max_{v\in \Delta_n} \sum_{t=1}^T \left\langle v_t-v,
      -\mathbf{A}^\intercal u_{t} \right\rangle\right] \\
    & \le 24 \lambda +24\lambda T \epsilon_{\G}^2 +
      \frac{\log \rbra*{mn}}{\lambda}+ 4T\epsilon_{\G}.
  \end{align*}
  Again using linearity of expectation and Markov's inequality, we conclude that
  with probability at least $2/3$
  \begin{equation*}
    T\rbra*{\max_{v\in \Delta_n}  \Ip{v}{\mbf{A}^{\trans}\hat{u}} -
      \min_{u\in \Delta_m} \Ip{u}{\mbf{A}\hat{v}}}\le 72 \lambda +
    \frac{3\log \rbra*{mn}}{\lambda} + 72T\lambda \epsilon^2_{\G}+
    12T\epsilon_{\G}.
  \end{equation*}
\end{proof}

\section{Consistent Quantum Amplitude Estimation}\label{append:consistent-amplitude-estimation}

\begin{thm}[Consistent phase estimation,~\cite{A2012,T2013}]\label{thm:consistent-phase-estimation}
  Suppose $U$ is a unitary operator.
  For every positive reals $\epsilon, \delta$, there is a quantum algorithm (a unitary quantum circuit)
  $\mathcal{A}$  such that, on input $O\rbra*{\log\rbra*{\epsilon^{-1}}}$-bit
  random string $s$, it holds that
  \begin{itemize}
    \item For every eigenvector $\ket{\psi_\theta}$ of $U$ (where
          $U \ket{\psi_\theta} = \exp\rbra*{\mathrm{i}\theta} \ket{\psi_\theta}$), with probability $\geq 1 - \epsilon$:
          \begin{equation*}
              \bra{\psi_{\theta}}\bra{f\rbra*{s,\theta}}\mathcal{A} \ket{\psi_\theta} \ket{0} \ge 1- \epsilon;
          \end{equation*}
          
    \item $f\rbra*{s, \theta}$ is a function of $s$ and $\theta$ such that
          $\abs{f\rbra*{s, \theta} - \theta} < \delta$,
  \end{itemize}
  with time complexity
  $\widetilde O\rbra*{\delta^{-1}} \cdot \poly\rbra*{\epsilon^{-1}}$.
\end{thm}

\begin{thm}[Consistent quantum amplitude estimation]\label{thm:consist-quantum-amp-est}
  Suppose $U$ is a unitary operator such that
  \begin{equation*}
    U \ket{0}_A \ket{0}_B = \sqrt{p} \ket{0}_A \ket{\phi_0}_B + \sqrt{1-p} \ket{1}_A \ket{\phi_1}_B.
  \end{equation*}
  where $p \in \sbra*{0, 1}$ and $\ket{\phi_0}$ and $\ket{\phi_1}$ are
  normalized pure quantum states.
  Then for every positive reals $\epsilon, \delta$, there is a quantum algorithm
  that, on input $O\rbra*{\log\rbra*{\epsilon^{-1}}}$-bit random string $s$,
  outputs $f\rbra*{s, p} \in \sbra*{0, 1}$ such that
  \begin{equation*}
    \Pr{\abs{f\rbra*{s, p} - p} \leq \delta } \geq 1 - \epsilon,
  \end{equation*}
  with time complexity
  $\widetilde O\rbra*{\delta^{-1}} \cdot \poly\rbra*{\epsilon^{-1}}$.
\end{thm}

\begin{proof}
  Suppose $U$ is a unitary operator such that
  \begin{equation*}
    U \ket{0}_A \ket{0}_B = \sqrt{p} \ket{0}_A \ket{\phi_0}_B +
    \sqrt{1-p} \ket{1}_A \ket{\phi_1}_B.
  \end{equation*}
  Let
  \begin{equation*}
    Q = - U \rbra*{I - 2 \ket{0}_A \bra{0} \otimes \ket{0}_B \bra{0}} U^\dag
    \rbra*{I - 2 \ket{0}_A \bra{0} \otimes I_B}.
  \end{equation*}
  Similar to the analysis in~\cite{BHMT02}, we have
  \begin{equation*}
    U\ket{0}_A \ket{0}_B = \frac{-\mathrm{i}}{\sqrt{2}}
     \bigl( \exp\rbra*{\mathrm{i}\theta_p} \ket{\psi_+}_{AB} -
      \exp\rbra*{-\mathrm{i}\theta_p} \ket{\psi_-}_{AB}\bigr),
  \end{equation*}
  where $\sin^2 \rbra*{\theta_p} = p$ ($0 \leq \theta_p < \pi/2$), and
  \begin{equation*}
    \ket{\psi_{\pm}}_{AB} = \frac{1}{\sqrt{2}}
    \rbra*{\ket{0}_A \ket{\phi_0}_B \pm
      \mathrm{i} \ket{1}_A \ket{\phi_1}_B }.
  \end{equation*}
  Note that $\ket{\psi_{\pm}}_{AB}$ are eigenvectors of $Q$, i.e.,
  $Q \ket{\psi_{\pm}}_{AB} = \exp\rbra*{\pm \mathrm{i} 2 \theta_p} \ket{\psi_{\pm}}_{AB}$.

  Now applying the algorithm $\mathcal{A}$ of consistent phase estimation of $Q$
  by~\cref{thm:consistent-phase-estimation} on state
  $U\ket{0}_A \ket{0}_B \otimes \ket{0}_C$ (with an
  $O\rbra*{\log\rbra*{\epsilon^{-1}}}$-bit random string $s$), we obtain
  \begin{equation*}
    \mathcal{A} \rbra*{U\ket{0}_A \ket{0}_B \otimes \ket{0}_C }
    = \frac{-\mathrm{i}}{\sqrt{2}} \Bigl( \exp\rbra*{\mathrm{i}\theta_p}
      \mathcal{A} \rbra*{\ket{\psi_+}_{AB} \ket{0 }_C } -
      \exp\rbra*{-\mathrm{i}\theta_p} \mathcal{A}
      \rbra*{\ket{\psi_-}_{AB} \ket{0 }_C } \Bigr).
  \end{equation*}
  Since each of $\ket{\psi_{\pm}}_{AB}$ is an eigenvector of $Q$, it holds that,
  with probability $\geq 1 - \epsilon$,
  \begin{equation*}
    \bra{\psi_{\pm}}_{AB}\bra{f\rbra*{s,\pm2\theta_p}}_C \mathcal{A}
    \rbra*{\ket{\psi_{\pm}}_{AB} \ket{0 }_C }\ge 1 - \epsilon.
  \end{equation*}
  which implies that
  $\mathcal{A} \rbra*{U\ket{0}_A \ket{0}_B \otimes \ket{0}_C }$ is
  $O\rbra*{\sqrt\epsilon}$-close to
  \begin{equation*}
    \frac{-\mathrm{i}}{\sqrt{2}} \Bigl( \exp\rbra*{\mathrm{i}\theta_p}
      \ket{\psi_+}_{AB} \ket{f\rbra*{s, 2\theta_p} }_C -
      \exp\rbra*{-\mathrm{i}\theta_p}
      \ket{\psi_-}_{AB} \ket{f\rbra*{s, -2\theta_p} }_C \Bigr)
  \end{equation*}
  in trace distance, where
  $\abs{f\rbra*{s, \pm 2 \theta_p} \mp 2 \theta_p} < \delta$.
  Measuring register $C$, we denote the outcome as $\gamma$, which will be
  either $f\rbra*{s, 2\theta_p}$ or $f\rbra*{s, -2\theta_p}$.
  Finally, output $\sin^2\rbra*{\gamma/2}$ as the estimate of $p$ (which is
  consistent).
  Since $\sin^2\rbra*{\cdot}$ is even and $2$-Lipschitz, the additive error is
  bounded by
  \begin{equation*}
    \abs{\sin^2\rbra*{\frac{\gamma}{2}} - p } \leq
    2 \abs{\abs{\frac{\gamma}{2}} - \abs{\theta_p} } < \delta.
  \end{equation*}
  Note that $\mathcal{A}$ makes
  $\widetilde O\rbra*{\delta^{-1}} \cdot \poly\rbra*{\epsilon^{-1}}$ queries to
  $Q$, thus our consistent amplitude estimation has quantum time complexity
  $\widetilde O\rbra*{\delta^{-1}} \cdot \poly\rbra*{\epsilon^{-1}}$.
\end{proof}

\begin{thm}[Error-Reduced Consistent quantum amplitude estimation]\label{thm:gapped-consist-quantum-amp-est}
  Suppose $U$ is a unitary operator such that
  \begin{equation*}
    U \ket{0}_A \ket{0}_B = \sqrt{p} \ket{0}_A \ket{\phi_0}_B +
    \sqrt{1-p} \ket{1}_A \ket{\phi_1}_B.
  \end{equation*}
  where $p \in \sbra*{0, 1}$ and $\ket{\phi_0}$ and $\ket{\phi_1}$ are
  normalized pure quantum states.
  Then for every positive integers $r$ and positive real $\delta$, there is a
  quantum algorithm that, on input $O\rbra*{r}$-bit random string $s$, outputs
  $f^*\rbra*{s, p} \in \sbra*{0, 1}$ such that
  \begin{equation*}
    \Pr{\abs{f^*\rbra*{s, p} - p} \leq \delta } \geq 1 - O\rbra*{\exp \rbra*{-r}},
  \end{equation*}
  with time complexity $\widetilde O\rbra*{\delta^{-1}} \cdot \poly\rbra*{r}$.
\end{thm}

\begin{proof}
  Consider that we divide the input random string $s$ into $r$ strings
  $s_1, s_2, \dots, s_r$ of length $O\rbra*{1}$.
  For each $i\in [r]$, we use \cref{thm:consist-quantum-amp-est} with
  input string $s_i$ and parameter $\epsilon = 1/10$.
  So we get, for each $i\in [r]$,
  \begin{equation*}
    \Pr{\abs{f\rbra*{s_i, p} - p} \leq \delta } \geq \frac{9}{10}.
  \end{equation*}
 Now we set $f^*(s, p)$ to be the median of the estimations $f(s_i,p)$ for
 $i\in [r]$.
 We claim it satisfies the desired property.
 To show that, we define random variables $X_i$ for $i\in [r]$ as follows:
 \begin{equation*}
   X_i =
   \begin{cases}
     1, & \textup{if } \abs{f\rbra*{s_i, p} - p} \leq \delta, \\
     0, & \textup{otherwise}.
   \end{cases}
 \end{equation*}
 Noticing $\E \left[\sum_{i=1}^r X_i\right]\ge 9r/10$, and by Chernoff bound, we
 have:
 \begin{equation*}
   \Pr{\sum_{i=1}^r X_i < \frac{r}{2}} \le \exp \rbra*{-\frac{8r}{45}}.
 \end{equation*}
 Thus with probability at least $1-\exp \rbra*{-{8r}/{45}}$, we know that at
 least half of the estimations fall into the interval
 $\sbra*{p-\delta, p+\delta}$, and then $f^*(s,p)$ returns a correct answer.
\end{proof}

\section{Details and Proofs of Fast Quantum Multi-Gibbs
  Sampling}\label{append:multi-Gibbs-proof}

We present the detailed version of the fast quantum multi-Gibss sampling.
Here, we use the shorthand $\Oracle_{p} = \Oracle_{p}\rbra{1,0}$, and it also
means the distribution of the sample.

We first define the notion of amplitude-encoding (a unitary operator that
encodes a vector in its amplitudes).

\begin{defi} [Amplitude-encoding] \label{def:amplitude-encoding} A unitary
  operator $V$ is said to be a $\beta$-amplitude-encoding of a vector
  $u \in \mr^{n}$ with non-negative entries, if
  \begin{equation*}
    \bra{0}_C V \ket{0}_C \ket{i}_A
    \ket{0}_B = \sqrt{\frac{u_i}{\beta}}\ket{i}_A\ket{\psi_i}_B
 \end{equation*}
 for all $i \in \sbra*{n}$.
\end{defi}

Then, as shown in \cref{algo:multi-Gibbs-full}, we can construct a quantum
multi-Gibbs sampler for a vector $u$ if an amplitude-encoding of the vector $u$
is given.
To complete the proof of \cref{thm:gibbs}, we only have to construct an
amplitude-encoding of $\mbf{A}z$ (see \cref{sec:gibbs-full} for details).

\begin{algorithm}[!htp]
  \caption{Quantum Multi-Gibbs Sampling implementing
    $\Oracle_{u}(k,\epsilon_{\G})$}\label{algo:multi-Gibbs-full}
  \begin{algorithmic}[1]
    \REQUIRE Sample count $k$, a $\beta$-amplitude-encoding $V$ of vector
    $u \in \mr^n$, polynomial $P_{2\beta} \in \mr\sbra*{x}$ that satisfies
    \cref{lem:poly-approx} with parameter
    $\epsilon_P = k\epsilon^2_{\G} / 300n $.
    \ENSURE $k$ independent samples $i_1, i_2, \dots, i_{k}$.
    \STATE Obtain
    $\mathcal{O}_{\tilde u} \colon \ket{i}\ket{0} \mapsto \ket{i}\ket{\tilde u_i}$
    using $\widetilde O\rbra*{\beta}$ queries to $V$, where
    $u_i \leq \tilde u_i \leq u_i + 1$, by consistent quantum amplitude
    estimation (\cref{thm:gapped-consist-quantum-amp-est}).
    \STATE Find the $k$ largest $\tilde u_i$'s by quantum $k$-maximum finding
    (\cref{thm:q-max}) and let $S$ be the set of their indexes.
    This can be done with $\widetilde O\rbra{\sqrt{nk}}$ queries to
    $\mathcal{O}_{\tilde u}$.
    \STATE Compute $\tilde u^* = \min\limits_{i \in S} \tilde u_i$, and
    $W = \rbra*{n-k} \exp\rbra*{\tilde u^*} + \sum\limits_{i \in S} \exp \rbra*{\tilde u_i}$.
    \FOR {$\ell = 1, \dots, k$} \STATE Prepare the quantum state
    \begin{equation*}
      \ket{u_{\text{guess}}} = \sum_{i \in S} \sqrt{\frac{\exp\rbra*{\tilde u_i}}{W} } \ket{i} +
      \sum_{i \notin S} \sqrt{\frac{\exp\rbra*{\tilde u^*}}{W} } \ket{i}.
    \end{equation*}
    \STATE Obtain
    $U_u = \rbra{V^\dag_{CAB} \otimes I_{D}} \rbra{V_{DAB}\otimes I_{C}}$
    being a block-encoding of $\diag\rbra*{u} / \beta$.
    Similarly, obtain $U^{\max}_{\tilde u}$ being a block-encoding of
    $\diag\rbra*{\max\cbra*{\tilde u, \tilde u^*}} / 2\beta$.
    \STATE Obtain $U^-$ being a block-encoding of
    $\diag\rbra*{u - \max\cbra*{\tilde u, \tilde u^*}} / 4\beta$ by the LCU
    (Linear-Combination-of-Unitaries) technique
    (\cref{thm:linear-combination-unitary}), using $O\rbra*{1}$ queries to $U_u$
    and $U^{\max}_{\tilde u}$.
    \STATE Obtain $U^{\exp}$ being a block-encoding of
    $P_{2\beta}\rbra*{\diag\rbra*{u - \max\cbra*{\tilde u, \tilde u^*}} / 4\beta}$
    by the QSVT technique (\cref{thm:qsvt}), using
    $O\rbra{\beta \log\rbra{\epsilon_P^{-1}}}$ queries to $U^-$.
    \STATE Post-select
    $\ket{\tilde u_{\text{post}}} = \bra{0}^{\otimes a} U^{\exp} \ket{u_\text{guess}} \ket{0}^{\otimes a}$
    by quantum amplitude amplification (\cref{thm:amplification}), and obtain
    $\ket{\tilde u_{\text{Gibbs}}} = \ket{\tilde u_{\text{post}}} / \Abs{\ket{\tilde u_{\text{post}}}}$.
    (Suppose $U^{\exp}$ has $a$ ancilla qubits.)
    \STATE Measure $\ket{\tilde u_{\text{Gibbs}}}$ in the computational basis
    and let $i_\ell \in \sbra*{n}$ be the outcome.
    \ENDFOR
    \STATE \textbf{Return} $i_1, i_2, \dots, i_{k}$.
  \end{algorithmic}
\end{algorithm}

\subsection{Useful Theorems}

\begin{thm}[Quantum state preparation,~\cite{GR2002,KP2017}]\label{thm:quantum-state-preparation}
  There is a data structure implemented on QRAM maintaining an array
  $a_1, a_2, \dots, a_{\ell}$ of positive numbers that supports the following
  operations.
  \begin{itemize}
    \item Initialization: For any value $c$, set $a_i \gets c$ for all
          $i \in \sbra*{\ell}$.
    \item Assignment: For any index $i$ and value $c$, set $a_i \gets c$.
    \item State Preparation: Prepare a quantum state
        \[
        \ket{a} = \sum_{i \in \sbra*{\ell}} \sqrt{\frac{a_i}{\Abs*{a}_1}} \ket{i}.
        \]
  \end{itemize}
  Each operation costs $\polylog\rbra*{\ell}$ time.
\end{thm}

\begin{thm}[Quantum $k$-maximum finding, Theorem 6
  of~\cite{DHHM2006}]\label{thm:q-max}
  Given $k \in \sbra*{n}$ and quantum oracle $\mathcal{O}_{u}$ for an array
  $u_1, u_2, \dots, u_n$, i.e., for every $i \in \sbra*{n}$,
  \begin{equation*}
    \mathcal{O}_{u} \ket{i} \ket{0} = \ket{i} \ket{u_i},
  \end{equation*}
  there is a quantum algorithm that, with probability $\geq 0.99$, finds a set
  $S \subseteq \sbra*{n}$ of cardinality $\abs{S} = k$ such that $u_i \geq u_j$
  for every $i \in S$ and $j \notin S$, using $O\rbra*{\sqrt{nk}}$ queries to
  $\mathcal{O}_u$.
\end{thm}

We now recall the definition of block-encoding, a crucial concept in quantum
singular value transformation~\cite{GSLW2019}, which is used in line $9$ to
$12$ in \cref{algo:multi-Gibbs-full}.

\begin{defi} [Block-encoding]\label{def:block-encoding}
  Suppose $A$ is a linear operator on $b$ qubits, $\alpha, \epsilon \geq 0$ and
  $a$ is a positive integer.
  A $\rbra*{b+a}$-qubit unitary operator $U$ is said to be an
  $\rbra*{\alpha, \epsilon}$-block-encoding of $A$, if
  \begin{equation*}
    \Abs*{\alpha \bra{0}^{\otimes a} U \ket{0}^{\otimes a} - A }_{\mathrm{op}} \leq \epsilon.
  \end{equation*}
\end{defi}

\begin{defi}[State Preparation Pair, Definition 28 of~\cite{GSLW2019}]\label{def:state-preparation-pair}
  Let $y\in \mr^n$ be a vector, specially in this context the number of
  coordinates starts from $0$.
  Suppose $\norm{y}_1\le \beta$.
  Let $\epsilon$ be a positive real.
  We call a pair of unitaries $\rbra*{P_L, P_R}$ acting on $b$ qubits a
  $\rbra*{\beta,\epsilon}$-state-preparation pair for $y$ if
  \begin{align*}
    P_L \ket{0}^{\otimes b} &= \sum_{j=0}^{2^b-1} c_j \ket{j},\\
    P_R \ket{0}^{\otimes b} &= \sum_{j=0}^{2^b-1} d_j \ket{j},
  \end{align*}
  such that:
  \begin{equation*}
   \sum_{j=0}^{m-1}\abs{\beta c_j^*d_j-y_j}\le \epsilon
  \end{equation*}
  and for $j\in [2^b]$, $j\ge m$, we require $c_j^* d_j =0$.
\end{defi}

We now state a theorem about linear combination of unitary operators,
introduced by~\cite{BCCKS2015} and~\cite{CW2012}.
The following form is from~\cite{GSLW2019}.
Again we restrict ourselves to the case of real linear combinations.

\begin{thm}[Linear Combination of Unitaries, Lemma 29 of~\cite{GSLW2019}]\label{thm:linear-combination-unitary}
  Let $\epsilon$ be a positive real number and $y\in \mr^n$ be a vector as in
  \cref{def:state-preparation-pair} with $\rbra*{\beta, \epsilon_1}$ state
  preparation pair $\rbra*{P_L , P_R}$.
  Let $\set{A_j}_{j=0}^{m-1}$ be a set of linear operators on $s$ qubits, and
  forall $j$, we have $U_j$ as an $(\alpha, \epsilon_2)$-block-encoding of $A_j$
  acting on $a+s$ qubits.
  Let
  \begin{equation*}
    W = \rbra*{\sum_{j=0}^{m-1}\ket{j}\bra{j}\otimes U_j} +
    \rbra*{I-\sum_{j=0}^{m-1}\ket{j}\bra{j}}\otimes I_{a+s},
  \end{equation*}
  Then we can implement a
  $(\alpha \beta, \alpha \epsilon_1 + \alpha \beta \epsilon_2)$-block-encoding
  of $A = \sum_{j= 0}^{m-1} y_j A_j$, with one query from $P_L^{\dagger}$, 
  $P_R$, and $W$.
\end{thm}

\begin{thm} [Eigenvalue transformation, Theorem 31
  of~\cite{GSLW2019}]\label{thm:qsvt}
  Suppose $U$ is an $\rbra*{\alpha, \epsilon}$-block-encoding of an Hermitian
  operator $A$.
  For every $\delta > 0$ and real polynomial $P \in \mr\sbra*{x}$ of degree $d$
  such that $\abs{P\rbra*{x}} \leq 1/2$ for all $x \in \sbra*{-1, 1}$, there is
  an efficiently computable quantum circuit $\tilde U$, which is a
  $\rbra*{1, 4d\sqrt{\epsilon/\alpha}+\delta}$-block-encoding of
  $P\rbra*{A/\alpha}$, using $O\rbra*{d}$ queries to $U$.
\end{thm}

Finally, for quantum amplitude amplification without knowing the exact value of
the amplitude, we need the following theorem:

\begin{thm} [Quantum amplitude amplification, Theorem 3
  of~\cite{BHMT02}]\label{thm:amplification}
  Suppose $U$ is a unitary operator such that
  \begin{equation*}
    U \ket{0}_A \ket{0}_B = \sqrt{p} \ket{0}_A \ket{\phi_0}_B +
    \sqrt{1-p} \ket{1}_A \ket{\phi_1}_B.
  \end{equation*}
  where $p \in \sbra*{0, 1}$ is unknown and $\ket{\phi_0}$ and $\ket{\phi_1}$
  are normalized pure quantum states.
  There is a quantum algorithm that outputs $\ket{0}_A\ket{\phi_0}_B$ with
  probability $\geq 0.99$, using $O\rbra*{1/\sqrt{p}}$ queries to $U$.
\end{thm}

\subsection{Main Proof} \label{sec:gibbs-full}

We generalize \cref{thm:gibbs} as follows.

\begin{thm} \label{thm:gibbs-full}
  \Cref{algo:multi-Gibbs-full} will produce $k$ independent and identical
  distributed samples from a distribution that is $\epsilon_{\G}$-close
  to $\Oracle_{u}$ in total variation distance, in quantum time
  $\widetilde O\rbra*{\beta \sqrt{nk}}$.
\end{thm}

It is immediate to show \cref{thm:gibbs} from \cref{thm:gibbs-full} by
constructing a $\beta$-amplitude-encoding $V$ of $\mbf{A}z$.
To see this, let $u = \mathbf{A}z$, then
$u_i = \rbra*{\mathbf{A} z}_i \in \sbra*{0, \beta}$.
By \cref{thm:quantum-state-preparation}, we can implement a unitary operator
$U_z^{\textup{QRAM}}$ such that
\begin{equation*}
  U_z^{\textup{QRAM}} \colon \ket{0}_C \ket{0}_B \mapsto
  \ket{0}_C \sum_{j \in \sbra*{n}} \sqrt{\frac{z_j}{\beta}} \ket{j}_B + \ket{1}_C \ket{\phi}_B.
\end{equation*}
Using two queries to $\mathcal{O}_{\mathbf{A}}$, we can construct a unitary
operator $\mathcal{O}_{\mathbf{A}}'$ such that
\begin{equation*}
  \mathcal{O}_{\mathbf{A}}' \colon \ket{0}_E \ket{i}_A \ket{j}_B
  \mapsto \rbra*{\sqrt{A_{i,j}} \ket{0}_E + \sqrt{1-{A_{i,j}}} \ket{1}_E} \ket{i}_A \ket{j}_B.
\end{equation*}
Let
\begin{equation} \label{eq:def-V-Az}
  V = \rbra*{\ket{0}_C\bra{0} \otimes \mathcal{O}_{\mathbf{A}}'
    + \ket{1}_C\bra{1} \otimes I_{EAB}} \rbra*{U_z^{\textup{QRAM}} \otimes I_{EA} }.
\end{equation}
It can be verified (see \cref{prop:V00i0}) that
\begin{equation*}
  \bra{0}_C \bra{0}_E V \ket{0}_C \ket{0}_E \ket{i}_A \ket{0}_B
  = \sum_{j \in \sbra*{n}} \sqrt{\frac{A_{i,j}z_j}{\beta} } \ket{i}_A \ket{j}_B,
\end{equation*}
and thus
$\bra{0}_C \bra{0}_E V \ket{0}_C \ket{0}_E \ket{i}_A \ket{0}_B =
\sqrt{u_i / \beta} \ket{i}_A \ket{\psi_i}_B$
for some $\ket{\psi_i}$.
Therefore, $V$ is a $\beta$-amplitude-encoding of $\mbf{A}z$.

Now, we will show \cref{thm:gibbs-full} in the following.

\begin{proof} [Proof of \cref{thm:gibbs-full}]
  Now we start to describe our algorithm.
  By our consistent quantum amplitude estimation
  (\cref{thm:gapped-consist-quantum-amp-est}), we choose an $O\rbra*{r}$-bit
  random string $s$, then we can obtain a quantum algorithm
  $\mathcal{O}_{\hat u}$ such that, with probability
  $1 - O\rbra*{\exp\rbra*{-r}}$, for every $i \in \sbra*{n}$, it computes
  $f^*\rbra*{s, u_i/\beta}$ with
  $\widetilde O\rbra*{\delta^{-1}} \cdot \poly\rbra*{r}$ queries to $V$,
  where $f^*\rbra*{s, p}$ is a function that only depends on $s$ and $p$, and it
  holds that
  \begin{equation*}
    \abs{f^*\rbra*{s, p} - p} \leq \delta
  \end{equation*}
  for every $p \in \sbra*{-1, 1}$.
  Here, $r, \delta$ are parameters to be determined.
  Note that 
  \begin{equation*}
    \frac{u_i}{\beta} = \Abs*{\bra{0}_C V\ket{0}_C \ket{i}_A \ket{0}_B}^2,
  \end{equation*}
  so when applying consistent quantum amplitude estimation, we just use a 
  controlled-XOR gate conditioned on the index
  and with $A$ the target system, before every query to $V$.
    
  By quantum $k$-maximum finding algorithm (\cref{thm:q-max}), we can find a set
  $S \subseteq \sbra*{n}$ with $\abs{S} = k$ such that
  $f^*\rbra*{s, u_i/\beta} \geq f^*\rbra*{s, u_j/\beta}$ for every $i \in S$ and
  $j \notin S$ with probability $0.99 - O\rbra*{\sqrt{nk} \exp\rbra*{-r}}$,
  using $O\rbra*{\sqrt{nk}}$ queries to $\mathcal{O}_{\hat u}$.
  To obtain a constant probability, it is sufficient to choose
  $r = \Theta\rbra*{\log\rbra*{n}}$.

  For each $i \in S$, again applying our consistent quantum amplitude estimation
  (\cref{thm:gapped-consist-quantum-amp-est}), we can obtain the value of
  $f^*\rbra*{s, u_i/\beta}$ with probability $1 - O\rbra*{\exp\rbra*{-r}}$,
  using $\widetilde O\rbra*{\delta^{-1}} \cdot \poly\rbra*{r}$ queries to $V$;
  then we set
  \begin{equation*}
    \hat u_i = \beta f^*\rbra*{s, \frac{u_i}{\beta}}
  \end{equation*}
  for all $i \in S$, with success probability $1 - O\rbra*{k\exp\rbra*{-r}}$ and
  using $\widetilde O\rbra*{k \delta^{-1}} \cdot \poly\rbra*{r}$ queries to $V$
  in total.
  It can be seen that $\abs{\hat u_i - u_i} \leq \beta\delta$ for every
  $i \in S$.
    
  Let $\tilde u_i = \hat u_i + \beta \delta$, and then we store $\tilde u_i$ for
  all $i \in S$ in the data structure as in \cref{thm:quantum-state-preparation}
  (which costs $O\rbra*{k}$ QRAM operations).
  Then, we calculate
  \begin{equation*}
    W = \rbra*{n-k} \exp\rbra*{\tilde u^*} + \sum_{i \in S} \exp \rbra*{\tilde u_i}
  \end{equation*}
  by classical computation in $\widetilde O\rbra*{k}$ time, where
  \begin{equation*}
    \tilde u^* = \min_{i \in S} \tilde u_i. 
  \end{equation*}
  By \cref{thm:quantum-state-preparation}, we can prepare the quantum state
  \begin{equation*}
    \ket{u_{\text{guess}}} = \sum_{i \in S}
    \sqrt{\frac{\exp\rbra*{\tilde u_i}}{W} } \ket{i} +
    \sum_{i \notin S} \sqrt{\frac{\exp\rbra*{\tilde u^*}}{W} } \ket{i}
  \end{equation*}
  in $\widetilde O\rbra*{1}$ time.
    
  Now we introduce another system $D$, and then let 
  \begin{equation*}
    U_u = \rbra{V^\dag_{CAB} \otimes I_{D}} \rbra{V_{DAB}\otimes I_{C}}. 
  \end{equation*}
  It can be shown (see \cref{prop:U_u}) that $U_u$ is a
  $\rbra*{1, 0}$-block-encoding of $\diag\rbra*{u} / \beta$.
  By QRAM access to $\tilde u_i$, we can implement a unitary operator
  \begin{equation*}
    V_{\tilde u} \colon \ket{i}_A \ket{0}_B \mapsto \ket{i}_A
    \rbra*{\sqrt{\frac{\max\cbra*{\tilde u_i, \tilde u^*}}{2\beta}}\ket{0}_B +
      \sqrt{1-\frac{\max\cbra*{\tilde u_i, \tilde u^*}}{2\beta}}\ket{1}_B }
  \end{equation*}
  in $\widetilde O\rbra*{1}$ time by noting that
  $\max\cbra*{\tilde u_i, \tilde u^*} = \tilde u_i$ if $i \in S$ and
  $\tilde u^*$ otherwise.
  We introduce one-qubit system $C$, and let
  \begin{equation*}
    U^{\max}_{\tilde u} = \rbra*{V_{\tilde u}^\dag \otimes I_C}
    \rbra*{\text{SWAP}_{BC} \otimes I_A} \rbra*{V_{\tilde u} \otimes I_C}.
  \end{equation*}
  It can be shown that $U^{\max}_{\tilde u}$ is a $\rbra*{1, 0}$-block-encoding
  of $\diag\rbra*{\max\cbra*{\tilde u, \tilde u^*}} / 2\beta$.
  Applying the LCU technique (\cref{thm:linear-combination-unitary}), we can
  obtain a unitary operator $U^-$ that is a $\rbra*{1, 0}$-block-encoding of
  $\diag\rbra*{u - \max\cbra*{\tilde u, \tilde u^*}} / 4\beta$, using
  $O\rbra*{1}$ queries to $U_u$ and $U^{\max}_{\tilde u}$.
  By the QSVT technique (\cref{thm:qsvt} and \cref{lem:poly-approx}), we can
  construct a unitary operator $U^{\exp}$ that is a
  $\rbra*{1, 0}$-block-encoding of
  $P_{2\beta}\rbra*{\diag\rbra*{u - \max\cbra*{\tilde u, \tilde u^*}} / 4\beta}$,
  using $O\rbra{\beta \log\rbra{\epsilon_P^{-1}}}$ queries to $U^-$, where
  \begin{equation*}
    \abs{P_{2\beta}\rbra*{x} - \frac 1 4 \exp\rbra*{2\beta x} } \leq \epsilon_{P}
  \end{equation*}
  for every $x \in \sbra*{-1, 0}$ and $\epsilon_P \in \rbra*{0, 1/2}$ is to be
  determined.
  Suppose $U^{\exp}$ has an $a$-qubit ancilla system, and let
  $\ket{\tilde u_{\text{post}}} = \bra{0}^{\otimes a} U^{\exp}
  \ket{u_\text{guess}} \ket{0}^{\otimes a}$.
  Note that
  \begin{equation*}
    \ket{\tilde u_{\text{post}}} = \sum_{i \in S} P_{2\beta}
    \rbra*{\frac{u_i - \tilde u_i}{4\beta}}
    \sqrt{\frac{\exp\rbra*{\tilde u_i}}{W}} \ket{i} +
    \sum_{i \notin S} P_{2\beta}\rbra*{\frac{u_i - \tilde u^*}{4\beta}}
    \sqrt{\frac{\exp\rbra*{\tilde u^*}}{W}} \ket{i}.
  \end{equation*}
  It can be shown (\cref{prop:u_post}) that
  $\Abs{\ket{\tilde u_{\text{post}}}}^2 \geq \Theta\rbra*{k / n}$; thus by
  quantum amplitude amplification (\cref{thm:amplification}), we can obtain
  \begin{equation*}
    \ket{\tilde u_{\text{Gibbs}}} = \frac{\ket{\tilde u_{\text{post}}} }
    {\Abs*{\ket{\tilde u_{\text{post}}}}}
  \end{equation*}
  using $O\rbra{\sqrt{n/k}}$ queries to $U^{\exp}$.
  By measuring $\ket{\tilde u_{\text{Gibbs}}}$ on the computational basis, we
  return the outcome as a sample from the distribution
  $\tilde u_{\text{Gibbs}}$; it can be shown (\cref{prop:diff-total-variation})
  that the total variation distance between $\tilde u_{\text{Gibbs}}$ and
  $\Oracle_{u}$ is bounded by
  \begin{equation*}
    d_{\text{TV}}\rbra*{\tilde u_{\text{Gibbs}}, \Oracle_{u}}
    \leq \sqrt{\frac{88n\epsilon_P}{k\exp\rbra*{-2\beta\delta} } }.
  \end{equation*}
  By taking $\delta = 1/2\beta$ and
  $\epsilon_P = k \epsilon^2_{\G}/300n$, we
  can produce one sample from $\tilde u_{\text{Gibbs}}$, using
  $\widetilde O\rbra{\beta \sqrt{n/k}}$ queries to $U_u$ and
  $U^{\max}_{\tilde u}$, with $\widetilde O\rbra{\beta \sqrt{nk}}$-time
  precomputation.

  Finally, by applying $k$ times the above procedure (with the precomputation
  processed only once), we can produce $k$ independent and identically
  distributed samples from $\tilde u_{\text{Gibbs}}$ that is
  $\epsilon_{\text{Gibbs}}$-close to the Gibbs distribution
  $\Oracle_{u}$, with total time complexity
  \begin{equation*}
    \widetilde O\rbra*{\beta \sqrt{nk}} + k \cdot
    \widetilde O\rbra*{\beta \sqrt{\frac{n}{k}}} = \widetilde O\rbra*{\beta \sqrt{nk}}.
  \end{equation*}
\end{proof}

\subsection{Technical Lemmas}

\begin{prop}\label{prop:V00i0}
  Let $V$ defined by \cref{eq:def-V-Az}, we have
  \begin{equation*}
    \bra{0}_C \bra{0}_D V \ket{0}_C \ket{0}_D \ket{i}_A \ket{0}_B =
    \sum_{j \in \sbra*{n}} \sqrt{\frac{A_{i,j}z_j}{\beta} } \ket{i}_A \ket{j}_B.
  \end{equation*}
\end{prop}
\begin{proof}
  \begin{align*}
    & V \ket{0}_C \ket{0}_D\ket{i}_A \ket{0}_B  \\
    = & \rbra*{\ket{0}_C\bra{0} \otimes \mathcal{O}_{\mathbf{A}}' + \ket{1}_C\bra{1}
        \otimes I_{AB}} \rbra*{\ket{0}_C\ket{0}_D\ket{i}_A \sum_{j \in \sbra*{n}}
        \sqrt{\frac{z_j}{\beta}} \ket{j}_B + \ket{1}_C\ket{0}_D\ket{i}_A\ket{\phi}_B} \\
    = & \ket{0}_C \sum_{j \in \sbra*{n}} \rbra*{\sqrt{A_{i,j}} \ket{0}_D +
        \sqrt{1-{A_{i,j}}} \ket{1}_D} \sqrt{\frac{z_j}{\beta}} \ket{i}_A \ket{j}_B +
        \ket{1}_C\ket{0}_D\ket{i}_A\ket{\phi}_B  \\
    = & \ket{0}_C \ket{0}_D \sum_{j \in \sbra*{n}} \sqrt{\frac{A_{i,j} z_j}{\beta}}
        \ket{i}_A \ket{j}_B + \ket{0}_C \ket{1}_D \sum_{j \in \sbra*{n}}
        \sqrt{\frac{\rbra*{1-{A_{i,j}}} z_j}{\beta}} \ket{i}_A \ket{j}_B +
        \ket{1}_C\ket{0}_D\ket{i}_A\ket{\phi}_B.
  \end{align*}
\end{proof}

\begin{prop}\label{prop:U_u}
  In the proof of \cref{thm:gibbs-full}, $U_u$ is a
  $\rbra*{1, 0}$-block-encoding of $\diag\rbra*{u} / \beta$.
\end{prop}
\begin{proof}
  To see this, for every $i, j \in \sbra*{n}$,
  \begin{align*}
    & \bra{0}_C \bra{0}_D \bra{j}_A \bra{0}_B U_u
       \ket{0}_C \ket{0}_D \ket{i}_A \ket{0}_B \\
    = & \bra{0}_C \bra{0}_D \bra{j}_A \bra{0}_B \rbra{V^\dag_{CAB} \otimes I_{D}}
        \rbra{V_{DAB}\otimes I_{C}}  \ket{0}_C \ket{0}_D \ket{i}_A \ket{0}_B \\
    = & \rbra*{\sqrt{u_j/\beta}\bra{0}_C\bra{0}_D
        \bra{j}_A\bra{\psi_j}_B + \bra{1}_C\bra{0}_D\bra{g_j}_{AB}}
        \rbra*{\sqrt{u_i/\beta}\ket{0}_C\ket{0}_D\ket{i}_A\ket{\psi_i}_B+
        \ket{0}_C\ket{1}_D \ket{g_i}_{AB}}\\
    = & \braket{j}{i}_A \frac{u_i}{\beta}.
  \end{align*}
\end{proof}

\begin{prop}\label{prop:u_post}
  In the proof of \cref{thm:gibbs-full}, if $\delta = 1/2\beta$,
  $E=\sum_{j\in \sbra*{n}} \exp\rbra*{u_{j}}$, and
  $\epsilon_P = k \epsilon^2_{\G}/300n$, then
  \begin{equation*}
    \Theta \rbra*{\frac{k}{n}} \leq \frac{E}{16W} - 2\epsilon_{P} \leq
    \Abs*{\ket{u_{\textup{post}}} }^2 \leq \frac{E}{16W} + 3\epsilon_P.
  \end{equation*}
\end{prop}

\begin{proof}
  We first give an upper bound for $W$ in terms of $u_i$ and $\tilde u^*$.
  Notice that $\tilde{u}_i \le u_i + 2 \beta\delta$ for all $i\in S$, we have:
  \begin{equation*}
    W  = \rbra*{n-k} \exp\rbra*{\tilde u^*} +
    \sum_{i \in S} \exp\rbra*{\tilde u_i}
    \leq \exp\rbra*{2\beta\delta} \rbra*{\rbra*{n-k} \exp\rbra*{u^*} +
      \sum_{i \in S} \exp\rbra*{u_i} }.
  \end{equation*}
  Note that
  \begin{align*}
    \frac{\rbra*{n-k} \exp\rbra*{u^*} + \sum\limits_{i \in S} \exp\rbra*{u_i} }
    {\sum\limits_{i \in \sbra*{n}} \exp\rbra*{u_i}} \leq \frac{n-k}{k} + 1 = \frac{n}{k},
  \end{align*}
  then we have
  \begin{equation}
    \label{eq:EW-bound}
    \frac{E}{W} \geq \sum_{i \in \sbra*{n}}
    \frac{\exp\rbra*{u_i}}{\exp\rbra*{2\beta\delta}
      \rbra*{\rbra*{n-k} \exp\rbra*{u^*} + \sum_{i \in S} \exp\rbra*{u_i} }}
    \geq \frac{k}{n} \exp\rbra*{-2\beta\delta}.
  \end{equation}
  With this, noting that $\rbra{a-b}^2\ge a^2-2ab$ for any real $a$ and $b$, we have
  \begin{align*}
    \Abs*{\ket{u_{\textup{post}}} }^2
    & = \sum_{i \in S} \rbra*{P_{2\beta}\rbra*{\frac{u_i - \tilde u_i}{4\beta}}}^2
      {\frac{\exp\rbra*{\tilde u_i}}{W}} + \sum_{i \notin S}
      \rbra*{P_{2\beta}\rbra*{\frac{u_i - \tilde u^*}{4\beta}}}^2
      {\frac{\exp\rbra*{\tilde u^*}}{W}} \\
    & \geq \sum_{i \in S} \rbra*{\rbra*{\frac 1 4 \exp \rbra*{\frac{u_i - \tilde u_i}{2}} }^2
      - 2\epsilon_P } {\frac{\exp\rbra*{\tilde u_i}}{W}} \\
    & \qquad +
      \sum_{i \notin S} \rbra*{\rbra*{\frac 1 4 \exp \rbra*{\frac{u_i - \tilde u^*}{2}}}^2
      - 2\epsilon_P } {\frac{\exp\rbra*{\tilde u^*}}{W}} \\
    & = \frac 1 {16} \rbra*{\sum_{i \in S}  \exp
      \rbra*{{u_i - \tilde u_i}}{\frac{\exp\rbra*{\tilde u_i}}{W}} +
      \sum_{i \notin S}  \exp \rbra*{{u_i - \tilde u^*}}
      {\frac{\exp\rbra*{\tilde u^*}}{W}} } \\
    & \qquad -
      2\epsilon_P \rbra*{\sum_{i \in S} \frac{\exp\rbra*{\tilde u_i}}{W} +
      \sum_{i \notin S} \frac{\exp\rbra*{\tilde u^*}}{W} } \\
    & \geq \frac{E}{16W} - 2\epsilon_P \\
    & \geq \Theta\rbra*{\frac{k}{n}}.
  \end{align*}
  On the other hand, a similar argument using the inequality
  $\rbra{a+b}^2\le a^2+3ab$ for positive real $a\ge b$ gives
  \begin{align*}
    \Abs*{\ket{u_{\textup{post}}} }^2
    & \leq \frac{E}{16W} + 3\epsilon_P.
  \end{align*}
  These yield the proof.
\end{proof}

\begin{prop}\label{prop:diff-total-variation}
  In the proof of \cref{thm:gibbs-full}, the total variation distance between
  the two distributions $\tilde u_{\text{Gibbs}}$ and $\Oracle_{u}$ is
  bounded by
  \begin{equation*}
    d_{\text{TV}}\rbra*{\tilde u_{\textup{Gibbs}}, \Oracle_{u}} \leq
    \sqrt{\frac{88n\epsilon_P}{k\exp\rbra*{-2\beta\delta} } }.
  \end{equation*}
\end{prop}

\begin{proof}
  Define $E = \sum\limits_{j\in \sbra*{n}} \exp\rbra*{u_{j}}$.
  Let
  \begin{equation*}
    \ket{u_{\text{Gibbs}}} = \sum_{i \in \sbra*{n}}
    \sqrt{\frac{\exp\rbra*{u_i}}{E}} \, \bigket{i}
  \end{equation*}
  be the intended quantum state with amplitudes the same as the Gibbs
  distribution $\Oracle_{u}$.
  The inner product between $\ket{\tilde u_{\text{post}}}$ and
  $\ket{u_{\text{Gibbs}}}$ can be bounded by:
  \begin{align*}
    \braket{\tilde u_{\text{post}}}{u_{\text{Gibbs}}}
    & = \sum_{i \in S} P_{2\beta}\rbra*{\frac{u_i - \tilde u_i}{4\beta}}
      \sqrt{\frac{\exp\rbra*{\tilde u_i}}{W}}
      \sqrt{\frac{\exp\rbra*{u_i}}{E}} \\
    & \qquad + \sum_{i \notin S} P_{2\beta}\rbra*{\frac{u_i - \tilde u^*}{4\beta}}
      \sqrt{\frac{\exp\rbra*{\tilde u^*}}{W}}
      \sqrt{\frac{\exp\rbra*{u_i}}{E}} \\
    & \geq \sum_{i \in S} \rbra*{\frac 1 4 \exp \rbra*{\frac{u_i - \tilde u_i}{2}} -
      \epsilon_P} \sqrt{\frac{\exp\rbra*{\tilde u_i}}{W}}
      \sqrt{\frac{\exp\rbra*{u_i}}{E}} \\
    & \qquad + \sum_{i \notin S} \rbra*{\frac 1 4 \exp
      \rbra*{\frac{u_i - \tilde u_i}{2}} - \epsilon_P}
      \sqrt{\frac{\exp\rbra*{\tilde u^*}}{W}}
      \sqrt{\frac{\exp\rbra*{u_i}}{E}} \\
    & \ge \frac{1}{4\sqrt{W E}} \biggl(\sum_{i \in \sbra*{n}}
      \exp\rbra*{u_i}\biggr) - \epsilon_P.
  \end{align*}
  The last step is by Cauchy's inequality.
  By \cref{prop:u_post} and \cref{eq:EW-bound}, we have
  \begin{align*}
    \abs{\braket{\tilde u_{\text{Gibbs}}}{u_{\text{Gibbs}}}}^2 =
    \frac{\abs{\braket{\tilde u_{\text{post}}}{u_{\text{Gibbs}}}}^2}
    {\norm{\ket{\tilde u_{\text{post}}}}^2}
    & \geq \frac{E }{16W \norm{\ket{\tilde u_{\text{post}}}}^2}
      - \frac{\epsilon_P }{2\norm{\ket{\tilde u_{\text{post}}}}^2}\\
    &\ge \frac{E }{16W \rbra*{\dfrac{E}{16W} + 3\epsilon_P}} -
      \frac{\epsilon_P}{2\norm{\ket{\tilde {u}_{\text{post}}}}^2} \\
    &\ge 1- \frac{48\epsilon_P}{E/W} -
      \frac{8\epsilon_P}{E/W - 32\epsilon_{P}} \\
    & \ge 1 - \frac{48n\epsilon_P}{k \exp\rbra*{-2\beta\delta} } -
      \frac{8n\epsilon_P}{k \exp\rbra*{-2\beta\delta} - 32n\epsilon_P} \\
    & \ge 1 - \frac{88n\epsilon_P}{k\exp\rbra*{-2\beta\delta} }.
  \end{align*}
  Finally, we have
  \begin{align*}
    d_{\text{TV}}\rbra*{\tilde u_{\textup{Gibbs}}, \Oracle_{u}}
    & \leq \frac 1 2 \tr \rbra*{\Big| \ket{\tilde u_{\text{Gibbs}}}
      \bra{\tilde u_{\text{Gibbs}}} - \ket{u_{\text{Gibbs}}}
      \bra{u_{\text{Gibbs}}} \Big| } \\
    & = \sqrt{1 - \abs{\braket{\tilde u_{\text{Gibbs}}}{u_{\text{Gibbs}}}}^2} \\
    & \leq \sqrt{\frac{88n\epsilon_P}{k\exp\rbra*{-2\beta\delta} } },
  \end{align*}
  which is bounded by $\epsilon_{\G}$ by the choice of $\epsilon_{P}$.
\end{proof}

\end{document}